\documentclass[a4paper,10pt]{article}

\usepackage{latexsym,amssymb,amsmath,amsfonts,epsf}
\usepackage[all]{xy}

\usepackage{bm}

\usepackage[normalem]{ulem}

\usepackage{amsthm}
\newtheorem{theorem}{Theorem}
\newtheorem{lemma}{Lemma}

\usepackage{color}
\usepackage{authblk}

\title{Invariant classification of second-order conformally flat superintegrable systems}
\author[1]{J.J. Capel\footnote{email:{\tt joshua.capel@gmail.com}joshu}}
\author[1]{J.M. Kress\footnote{email:\tt j.kress@unsw.edu.au}}
\affil[1]{School of Mathematics and Statistics\\University of New South Wales\\Sydney NSW 2052, Australia}

\begin{document}

\maketitle

\begin{abstract}
In this paper we continue the work of Kalnins \textit{et al} in classifying all second-order conformally-superintegrable (Laplace-type) systems over
conformally flat spaces, using tools from algebraic geometry and classical invariant theory.
The results obtained show, through St\"ackel equivalence, that the list of known nondegenerate superintegrable systems over three-dimensional conformally flat spaces is complete.
In particular, a 7-dimensional manifold is determined such that each point corresponds to a conformal class of superintegrable systems.
This manifold is foliated by the nonlinear action of the conformal group in three-dimensions.
Two systems lie in the same conformal class if and only if they lie in the same leaf of the foliation.
This foliation is explicitly described using algebraic varieties formed from representations of the conformal group.
The proof of these results rely heavily on Gr\"obner basis calculations using the computer algebra software packages Maple and Singular.
\end{abstract}

\section{Introduction}

The classification of second-order superintegrable systems, both in classical and quantum mechanics, is a topic which is often revisited \cite{kalnins2000completenessE2C,kalnins2000completeness2sphere,kalnins2001completeness,kalnins2002superintegrability,kalnins2006classification,kress2007equivalence,kalnins2007algebraicvarieties2d,kalnins2007fine,kalnins2008fine}, usually with the goal of increasing our understanding of superintegrable systems in general.
Superintegrablility has deep links to Quasi-exact solvability, and it has been conjecture that all Quasi-exact solvable systems arise from superintegrable systems.
In the area of special functions, connections have been demonstrated between the contraction of second-order superintegrable systems and the Askey-Wilson scheme for orthogonal hypergeometric polynomials \cite{kalnins2014contractions}.

The classification of two-dimensional second-order systems is complete \cite{kalnins2001completeness,kalnins(2005):053510,kalnins(2006):093501} and substantial steps have been taken towards the classification of three-dimensional second-order systems \cite{kalnins2008fine,kalnins2011laplace}.
Most of the results obtained have relied on the use of separation of variables,
however some recent work in this area has focused on classifying systems based on the structure of their symmetry algebra \cite{kress2007equivalence,daskaloyannis2010quadratic,marquette2012classical} or by investigating their integrability conditions \cite{kalnins2008fine,kalnins2011laplace,kalnins2007algebraicvarieties3d}.
Both of these approaches may provide techniques that can be applied to non-separable systems.

This classification work naturally leads to consideration of so-called conformally-superintegrable (Laplace-type) systems \cite{kalnins2011laplace}, where the Poisson bracket of the symmetries with the Hamiltonian are only guaranteed to vanishes on the zero-energy hypersurface.
The link between superintegrable and conformally-superintegrable systems is given by the St\"ackel transform \cite{boyer:778}.
Using this transform, statements about superintegrable systems can be translated into statements about conformally-superintegrable systems.

The two-dimensional second-order maximally superintegrable maximum-parameter (non-degenerate) systems are St\"ackel equivalent to systems over constant curvature spaces \cite{kalnins(2005):053510}, and these have been completely classified \cite{kalnins2001completeness}.
The equivalence classes are identified uniquely by the structure of their symmetry algebras, which close polynomially \cite{kress2007equivalence}.
These systems were originally classified by the coordinates in which they separate, however in the case of Euclidean systems an alternative classification technique was found by foliating the algebraic variety of integrability conditions under the action of the Euclidean group \cite{kalnins2007algebraicvarieties2d}.

Similarly, the three-dimensional maximum-parameter systems have closed symmetry algebras distinct to each St\"ackel class \cite{daskaloyannis2010quadratic}.
Unfortunately this statement does not extend to non maximum-parameter systems (degenerate systems) as their symmetry algebra is no longer guaranteed to close polynomially.
However, foliating the variety of integrability conditions into subvarieties is possible in both the maximum-parameter and non maximum-parameter cases, and is currently being investigated \cite{kalnins2007fine,kalnins2011laplace,kalnins2007algebraicvarieties3d}.
The main result in this paper is the description of such a foliation, that provides a complete classification for the three-dimensional maximum-parameter case.
These results are achieved without having to resort to separation of variables and it is clear how such techniques could be applied to the higher-dimensional second-order systems.
It is envisaged that the techniques used could also be applied to systems with higher-order symmetries, for which no general structure results are known.

\subsection{Varieties Classifying Superintegrable Systems}

The use of algebraic varieties and algebraic integrability conditions has appeared in the study of integrable and superintegrable systems. For example,
the space of Killing tensors on a constant sectional curvature manifold has a natural description as an algebraic variety \cite{schobel2012variety,schobel2012algebraic}.

The use of algebraic varieties in this paper mirrors that of Ref.~\cite{kalnins2007algebraicvarieties2d} where the two-dimensional second-order Euclidean superintegrable systems were reclassified using the integrability conditions to create an algebraic variety such that \emph{``each point on the variety corresponds to a superintegrable system.
The Euclidean group $E(2,C)$ acts on the variety such that two points determine the same superintegrable system if and only if they lie on the same leaf of the foliation''}.
Specifically, the potentials for each system are governed by a system of linear PDEs, with unspecified coefficient functions.
The integrability conditions for the coefficient functions in this system of second-order linear PDEs determine the algebraic variety, which is foliated by the action of the Euclidean group.
The leaves of the foliation, i.e.\ the equivalence classes, are described by a set of polynomial ideals.
As a result only polynomial evaluations are necessary to determine the equivalence class of a given system.
This approach is conceptually simple and has been put to use classifying three-dimensional Euclidean systems \cite{kalnins2007algebraicvarieties3d}.
However, a complete list of classifying polynomial ideals was not given and the proof still relied on separation of variables.

The techniques described above were revisited in Ref.~\cite{kalnins2011laplace} in the study of conformally-superintegrable (Laplace-type) systems over flat space.
It was remarked in that paper that the classification of three-dimensional maximum-parameter Laplace-type systems could be accomplished by studying how the non-linear action of the conformal group foliates a particular 13-dimensional manifold $\mathbb{C}^{13}$.
It was also suggested that, since all systems are St\"ackel equivalent to a superintegrable system on either the plane or the sphere, the classification could be completed by specialising a term in the potential to be the conformal factor for a constant curvature metric.
However, in this paper it is shown that the classification can be completed without specifying a particular metric.
It had been noted earlier that the 13-dimensional manifold splits naturally into the direct sum $\mathbb{C}^3\oplus\mathbb{C}^3\oplus\mathbb{C}^7$, but the fact that the conformal group acts transitively and independently on the 6-dimensional component $\mathbb{C}^3\oplus\mathbb{C}^3$ was overlooked.
With this observation the task can be simplified to just classifying the orbits on the remaining component $\mathbb{C}^7$.

Under the action of a conformal change of variables, the aforementioned 7-dimensional manifold, coming initially from an $SO(3,\mathbb{C})$ representation, can be considered to be the coefficients of a sextic polynomial in one variable acted on by fractional-linear transformations.
Basic results from classical invariant theory show that the orbits under this action are uniquely identified by two types of absolute invariants: the root multiplicities of the sextic and the multi-ratios\footnote{The multi-ratio is equivalent to the standard cross-ratio by a permutation of the indices.} between the roots.
The polynomial identified with our $SO(3,\mathbb{C})$ representation is determined locally up to projective linear transformations, essentially as an element of $\mathbb{CP}^6$.
Translation of the regular point allows us to move between projectively inequivalent sextics and the full classification is achieved by determining the orbits of this action as Zariski open subsets of algebraic varieties, and involves a mixture of classical invariant theory and algebraic-geometry.
The classification that is obtained proves the list of known maximum-parameter potentials is complete.

The calculations in this paper were done using Gr\"obner bases with the assistance of the computer algebra packages Maple\footnote{"Maple is a trademark of Waterloo Maple Inc."} and Singular\cite{DGPS}.

\section{Conformally-Superintegrable Systems}\label{section:conformally_superintegrable_systems}

We consider a classical system with a natural Hamiltonian of the form
\[
 H=\frac{p_{x_1}^2+p_{x_2}^2+p_{x_3}^2}{\lambda(\mathbf{x})}+V(\mathbf{x})
\]
where the coordinates are $\mathbf{x}=(x_1,x_2,x_3)$ and the generalised momenta are $\mathbf{p}=(p_{x_1},p_{x_2},p_{x_3})$.
Following the work of Ref.~\cite{kalnins2011laplace}, we are interested in systems possessing second-order conformal symmetries; that is, functions of the form
\begin{equation}
 L(\mathbf{x},\mathbf{p})=\sum_{i,j=1}^{3} a^{ij}(\mathbf{x})p_{x_i}p_{x_j} + W(\mathbf{x})\label{eq:the_conformal_constant_L}
\end{equation}
such that the Poisson bracket of the Hamiltonian $H$ with a conformal symmetry $L$ gives 
\[
 \left\{H,L\right\}_{PB}=\rho_{L} H,
\]
where $\rho_{L}(\mathbf{x},\mathbf{p})$ is polynomial in the momentum.
In the case where $L$ is a second-order conformal symmetry, $\rho_{L}$ will necessarily be linear in the momentum.
Additionally, if $\rho_L\equiv0$ then $L$ is a will be called a \emph{true symmetry} of $H$.

Any function of the form $R(\mathbf{x},\mathbf{p})H(\mathbf{x},\mathbf{p})$ is a conformal symmetry of $H$ and so to avoid these trivialities, a conformal symmetry will only be identified up to the addition of a multiple of the Hamiltonian.
It's worth noting that this identification makes the Hamiltonian $H$ equivalent to zero.

We will call an $n$-dimensional system \emph{maximally} conformally superintegrable if it possesses $2n-2$  functionally-independent and inequivalent conformal symmetries.

Along the hypersurface $H=0$ these $2n-2$ conformal symmetries become true-symmetries and thus are constant along the trajectories.
Including the Hamiltonian, these $2n-1$ constants of the motion will be sufficient to solve for $(\mathbf{x},\mathbf{p})$ analytically as a one-parameter trajectory.

\subsection{St\"ackel Equivalence (Conformal Classes of systems)}

The classification of conformally-superintegrable systems can be simplified by only considering Hamiltonians which are not equivalent up to a scaling factor.
This means considering St\"ackel classes of conformally equivalent systems.

\begin{lemma}\label{lemma:all_conformally_superintegrable_sytems_are_over_flat_space}
 Every conformally-superintegrable system over a conformally-flat space can be conformally scaled to a conformally-superintegrable system over flat space.
\end{lemma}
\begin{proof}
 By hypothesis the Hamiltonian is of the form
 \begin{equation}
  H=\frac{p_{x_1}^2+p_{x_2}^2+p_{x_3}^2}{\lambda}+V,
 \end{equation}
and possesses conformal constants $L$ which satisfy
 \begin{equation}
  \left\{H,L\right\}_{PB}=\rho_L H.
 \end{equation}

Scaling the Hamiltonian by the conformal factor $\lambda$ gives the new Hamiltonian
\begin{align}
 \widetilde{H}=\lambda H=p_{x_1}^2+p_{x_2}^2+p_{x_3}^2+\lambda V.\label{eq:conformal_scaling_of_the_hamiltonian}
\end{align}
The Poisson-Bracket of this new Hamiltonian and the original conformal symmetries is
 \begin{align}
 \left\{\widetilde{H} ,L\right\}_{PB}
         &=\left\{\lambda H ,L\right\}_{PB} \nonumber\\
         &=\lambda\left\{ H ,L\right\}_{PB}+H\left\{\lambda ,L\right\}_{PB} \nonumber\\
         &=\lambda\rho_L H+H\left\{\lambda ,L\right\}_{PB} \nonumber\\
         &=\left(\rho_L +\frac{\left\{\lambda ,L\right\}_{PB}}{\lambda}\right)\lambda H \nonumber\\
         &=\left(\rho_L +\frac{\left\{\lambda ,L\right\}_{PB}}{\lambda}\right)\widetilde{H}.
 \end{align}
The factor 
\[
  \widetilde{\rho}_L=\rho_L +\frac{\left\{\lambda ,L\right\}_{PB}}{\lambda}
\]
is polynomial in the momentum and hence the $L$'s are also a conformal symmetries of $\widetilde{H}$, clearly a Hamiltonian over flat space.
\end{proof}

So henceforth the conformally-superintegrable systems will be assumed to be over flat space, conformally scaling if necessary.
There do exist superintegrable systems which are not St\"ackel equivalent to conformally flat systems \cite{kalnins2013nonconformallyflatsuperintegrability}, but these are not the subject of this paper. 

Note that a superintegrable system is also a conformally-superintegrable system with conformal symmetries $L$ for which $\rho_L=0$, so an immediate consequence of lemma~\ref{lemma:all_conformally_superintegrable_sytems_are_over_flat_space} is that any superintegrable system over a conformally-flat space is equivalent (by conformal scaling) to a conformally-superintegrable one on flat space.

The following theorem shows this procedure can be reversed, specifically a conformally-superintegrable system
can be taken to a superintegrable system, if the Hamiltonian is rescaled by a conformal factor that is chosen to be
a term in the potential. In the context of second-order superintegrable systems this procedure is the known as
either the coupling-constant-metamorphosis (CCM) or the St\"ackel transform. A discussion of similarities 
(and more importantly, differences) between the CCM and the St\"ackel transform can be found in Ref.~\cite{post:265}.

\begin{theorem}\label{thm:conformally_superintegrable_to_superintegrable_via_staeckel_transform}
 If $H=H_0+\alpha U$ is a Hamiltonian with a parameter $\alpha$, and with a conformal symmetry $L(\alpha)=L_0+\alpha W_U$, then the new Hamiltonian $\widetilde{H}=\dfrac{H}{U}$ has true symmetry $L(-\widetilde{H})$. 
\end{theorem}

\begin{proof} 
  This proof is almost identical to the proof of theorem~1 from Ref.~\cite{1751-8121-43-3-035202}.
  Firstly note that, given functions of the form $G(\mathbf{x},\mathbf{p})$ and $F(a,\mathbf{x},\mathbf{p})$, where $a=\tau(\mathbf{x},\mathbf{p})$, then
  \begin{align}
     \{F,G\}=
          \left[\{F(a,\mathbf{x},\mathbf{p}),G(\mathbf{x},\mathbf{p})\}\right]_{a=\tau(\mathbf{x},\mathbf{p})}
         +\left[\partial_a F(a,\mathbf{x},\mathbf{p})\right]_{a=\tau(\mathbf{x},\mathbf{p})}\{\tau(\mathbf{x},\mathbf{p}),G(\mathbf{x},\mathbf{p})\}.
         \label{eq:chain_rule_for_the_poisson_braket}
  \end{align}
  Consider the conformal symmetry $L$ which, by hypothesis, satisfies a relation of the form
  \[
     \{H+\alpha U, L(\alpha)\}=\rho(H+\alpha U),
  \]
  and so
  \begin{align*}
     \{H,L(\alpha)\} &= -\alpha\{U,L(\alpha)\}+\rho(H+\alpha U).
  \end{align*}
  Using these it can be shown
  \begin{align}
     \{\widetilde{H},L(\alpha)\}
        &= \left\{\frac{H}{U},L(\alpha)\right\} \nonumber\\
        &= -\frac{H}{U^2}\left\{U,L(\alpha)\right\}+\frac{\left\{H,L(\alpha)\right\}}{U} \nonumber\\
        &= -\frac{H}{U^2}\left\{U,L(\alpha)\right\}+\frac{-\alpha{U,L(\alpha)}+\rho(H+\alpha U)}{U} \nonumber\\
        &= -\frac{\widetilde{H}+\alpha}{U}\left\{U,L(\alpha)\right\}+\frac{\rho(H+\alpha U)}{U}.\label{eq:poisson_commutator_of_the_conformally_rescaled_hamiltonian_and_the_original_constant}
  \end{align}
  So, using \eqref{eq:chain_rule_for_the_poisson_braket} and \eqref{eq:poisson_commutator_of_the_conformally_rescaled_hamiltonian_and_the_original_constant}, we find
  \begin{align*}
    \{\widetilde{H},L(-\widetilde{H})\}
       &=  \left[ \partial_\alpha L(\alpha)\left\{\widetilde{H},\widetilde{H}\right\}
                 -\frac{\widetilde{H}+\alpha}{U}\left\{U,L(\alpha)\right\}
                 +\frac{\rho(H+\alpha U)}{U}
           \right]_{\alpha=-\widetilde{H}} \\
       &= \frac{\rho(H-\widetilde{H}U)}{U} \\
       &= 0.
  \end{align*}
  Thus $L(-\widetilde{H})$ is a true symmetry for the transformed Hamiltonian.
\end{proof}

The proof of theorem~\ref{thm:conformally_superintegrable_to_superintegrable_via_staeckel_transform} is almost the same as the proof of the St\"ackel transform between superintegrable systems, the subtle difference however is the reason behind $\rho(H-\widetilde{H}U)$ vanishing.

Note that the conformal rescaling makes the parameter $\alpha$ into an additive constant, such an additive constant was lacking from the definition of a conformally-superintegrable system.
Also note that if $L(\alpha)$ is a trivial conformal symmetry of the form $F(\mathbf{x},\mathbf{p})(H_0+\alpha U)$ then $L(-\widetilde{H})=0$.

\subsection{Maximum parameter (Nondegenerate) Potentials}

\subsubsection{Bertrand-Darboux Equations}

As was shown by Kalnins \textit{et at} \cite{kalnins(2005):103507,kalnins2011laplace} any potential with 5 
functionally-linearly-independent second-order conformal symmetries (including the Hamiltonian) satisfies a set of linear PDEs of the form
\begin{align}
 \frac{\partial^2 V}{\partial x_2^2} &=
                                  \frac{\partial^2 V}{\partial x_1^2}
         + A^{22}_{1}(\mathbf{x}) \frac{\partial   V}{\partial x_1  }
         + A^{22}_{2}(\mathbf{x}) \frac{\partial   V}{\partial x_2  }
         + A^{22}_{3}(\mathbf{x}) \frac{\partial   V}{\partial x_3  }
         + A^{22}_{0}(\mathbf{x}) V,\nonumber \\
 \frac{\partial^2 V}{\partial x_3^2} &=
                                  \frac{\partial^2 V}{\partial x_1^2}
         + A^{33}_{1}(\mathbf{x}) \frac{\partial   V}{\partial x_1  }
         + A^{33}_{2}(\mathbf{x}) \frac{\partial   V}{\partial x_2  }
         + A^{33}_{3}(\mathbf{x}) \frac{\partial   V}{\partial x_3  }
         + A^{33}_{0}(\mathbf{x}) V,\nonumber \\
 \frac{\partial^2 V}{\partial x_1\partial x_2} &=
           \hphantom{\frac{\partial^2 V}{\partial x_1^2}+\vphantom{a}}
           A^{12}_{1}(\mathbf{x}) \frac{\partial   V}{\partial x_1  }
         + A^{12}_{2}(\mathbf{x}) \frac{\partial   V}{\partial x_2  }
         + A^{12}_{3}(\mathbf{x}) \frac{\partial   V}{\partial x_3  }
         + A^{12}_{0}(\mathbf{x}) V,\nonumber \\
 \frac{\partial^2 V}{\partial x_1\partial x_3} &=
           \hphantom{\frac{\partial^2 V}{\partial x_1^2}+\vphantom{a}}
           A^{13}_{1}(\mathbf{x}) \frac{\partial   V}{\partial x_1  }
         + A^{13}_{2}(\mathbf{x}) \frac{\partial   V}{\partial x_2  }
         + A^{13}_{3}(\mathbf{x}) \frac{\partial   V}{\partial x_3  }
         + A^{13}_{0}(\mathbf{x}) V,\nonumber \\
 \frac{\partial^2 V}{\partial x_2\partial x_3} &=
           \hphantom{\frac{\partial^2 V}{\partial x_1^2}+\vphantom{a}}
           A^{23}_{1}(\mathbf{x}) \frac{\partial   V}{\partial x_1  }
         + A^{23}_{2}(\mathbf{x}) \frac{\partial   V}{\partial x_2  }
         + A^{23}_{3}(\mathbf{x}) \frac{\partial   V}{\partial x_3  }
         + A^{23}_{0}(\mathbf{x}) V.\label{eq:the_PDE_governing_V}
\end{align}
The systems being studied will be assumed to depends on five parameters (maximum-parameter), meaning the values of $\frac{\partial^2 V}{\partial x_1^2}, \frac{\partial V}{\partial x_1}, \frac{\partial V}{\partial x_2}, \frac{\partial V}{\partial x_3}$ and $V$ can be freely specified at any regular point in the system.
The maximum-parameter assumption allowed Kalnins \textit{et at} to prove their ${(4 \implies 5)}$ theorem, showing that there are sufficiently many second-order conformal symmetries to allow for the value of $a^{ij}$'s in \eqref{eq:the_conformal_constant_L} to be specified arbitrarily.
This freedom puts fairly tight integrability conditions on the coefficient functions $A^{ij}_k$ in \eqref{eq:the_PDE_governing_V}, and allows the five $A^{ij}_0$ to be written as quadratics in the $A^{ij}_k, k\neq0$.
The exact expression for these quadratics can be found in Ref.~\cite{kalnins2011laplace}, but their exact form is unimportant for this analysis.
The remaining fifteen variables $A^{ij}_k, k\neq0$ can be shown to satisfy the following 5 linear equations
\begin{align}
  A^{22}_1-A^{12}_2 &= A^{33}_1-A^{13}_3, \nonumber\\
  A^{23}_3-A^{12}_1 &= A^{33}_2, \nonumber\\
  A^{23}_2-A^{13}_1 &= A^{22}_3, \nonumber\\
           A^{12}_3 &= A^{13}_2, \nonumber\\
           A^{12}_3 &= A^{23}_1. 
\end{align}
These reduce our fifteen coefficient functions down to a set of ten, which, for the sake of symmetry, can be parameterised as follows
\begin{equation}
  \begin{array}{lll}
    A^{22}_1 = 6 S^1+2 R^{12}_2+R^{13}_3,  & A^{22}_2 = -6 S^2-2 R^{12}_1-R^{23}_3,   & A^{22}_3 = R^{23}_2-R^{13}_1,           \\
    A^{33}_1 = 6 S^1+R^{12}_2+2 R^{13}_3,  & A^{33}_2 = -R^{12}_1+R^{23}_3,           & A^{33}_3 = -6 S^3-R^{23}_2-2 R^{13}_1,  \\
    A^{12}_1 = R^{12}_1-3 S^2,             & A^{12}_2 = R^{12}_2-3 S^1,               & A^{12}_3 = Q^{123},                     \\
    A^{13}_1 = R^{13}_1-3 S^3,             & A^{13}_2 = Q^{123},                      & A^{13}_3 = R^{13}_3-3 S^1,              \\
    A^{23}_1 = Q^{123},                    & A^{23}_2 = R^{23}_2-3 S^3,               & A^{23}_3 = R^{23}_3-3 S^2.
  \end{array}
\end{equation}
If the 10 new variable names
\begin{align}
(S^{1},S^{2},S^{3},R^{12}_{1},R^{12}_{2},R^{13}_{1},R^{13}_{3},R^{23}_{2},R^{23}_{3},Q^{123}), \label{QRS_variable_names}
\end{align}
are considered symmetric in the raised indices, then a permutation of the coordinates $x_i$ just corresponds to an equivalent permutation of the indices in $S^i,R^{ij}_i,Q^{123}$, a symmetry that the coefficient functions in equations \eqref{eq:the_PDE_governing_V} were lacking.
Another motivation for this parameter choice is, as will be shown in section \eqref{sec:rotation_apapted_variables}, the sets of variables $\{\mathbf{Q},\mathbf{R}\}$ and $\{\mathbf{S}\}$ each carry a separate irreducible representations of the $SO(3,\mathbb{C})$ Lie group.

The rest of the integrability conditions allow all the derivatives of the $Q,R,S$'s to be calculated.
These derivatives are quadratic in $Q,R,S$ and a sample of which are
\begin{multline}\label{eq:dR12_1/dx1}
 \frac{\partial {R}^{12}_1}{\partial x_1} = -\frac{2}{3} {R}^{12}_2 {R}^{23}_3+\frac{2}{3} {R}^{13}_3 {R}^{23}_3+\frac{4}{3} {Q}^{123} {R}^{23}_2+\frac{5}{3} {Q}^{123} {R}^{13}_1\\
-{R}^{12}_1 {R}^{13}_3-{R}^{12}_1 {S}^{1}+\left({R}^{13}_3+3 {R}^{12}_2\right) {S}^{2} +2 {Q}^{123} {S}^{3},
\end{multline}
\begin{multline}\label{eq:dR12_1/dx2}
  \frac{\partial {R}^{12}_1}{\partial x_2} = \frac{3}{5} {R}^{12}_2 {R}^{13}_3-\frac{1}{15} {R}^{13}_1 {R}^{23}_2-\frac{11}{15} {R}^{12}_1 {R}^{23}_3 \\
 +\frac{8}{15} \left({R}^{13}_1\right)^2+\frac{1}{5} \left({R}^{12}_1\right)^2-\frac{4}{5} \left({R}^{23}_2\right)^2+\frac{8}{15} \left({R}^{13}_3\right)^2 +\frac{1}{5} \left({R}^{12}_2\right)^2-\frac{4}{5} \left({R}^{23}_3\right)^2\\
 +\frac{2}{15} \left({Q}^{123}\right)^2
 -\left({R}^{13}_3+3 {R}^{12}_2\right) {S}^{1}-{R}^{12}_1 {S}^{2}+{R}^{13}_1 {S}^{3},
\end{multline}
\begin{multline}\label{eq:dR12_1/dx3}
\frac{\partial {R}^{12}_1}{\partial x_3} =
    -\frac{1}{3} {Q}^{123} {R}^{12}_2
        -\frac{1}{3} {Q}^{123} {R}^{13}_3
        +\frac{1}{3} {R}^{23}_2 {R}^{12}_1
        +\frac{1}{3} {R}^{23}_3 {R}^{13}_1 \\
    -2 {Q}^{123} {S}^{1}
        -{R}^{13}_1 {S}^{2}
        -{R}^{12}_1 {S}^{3},
\end{multline}
\begin{multline}\label{eq:dS1/dx1}
 \frac{\partial {S}^{1}}{\partial x_1} = -\frac{17}{90} {R}^{12}_2 {R}^{13}_3+\frac{1}{30} {R}^{13}_1 {R}^{23}_2+\frac{1}{30} {R}^{12}_1 {R}^{23}_3 \\
-\frac{7}{45} \left({R}^{13}_3\right)^2 +\frac{1}{15}\left({R}^{23}_3\right)^2-\frac{7}{45}\left({R}^{12}_1\right)^2-\frac{11}{90} \left({Q}^{123}\right)^2-\frac{7}{45} \left({R}^{13}_1\right)^2-\frac{7}{45}\left( {R}^{12}_2\right)^2 \\
+\frac{1}{15}\left({R}^{23}_2\right)^2+\frac{1}{2}\left({S}^{2}\right)^2+\frac{1}{2} \left({S}^{3}\right)^2 -\frac{1}{2} \left({S}^{1}\right)^2,
\end{multline}
\begin{multline}\label{eq:dS1/dx2}
 \frac{\partial {S}^{1}}{\partial x_2} = -\frac{1}{9} {R}^{13}_3 {R}^{23}_3-\frac{2}{9} {Q}^{123} {R}^{23}_2+\frac{1}{9} {R}^{12}_1 {R}^{13}_3+\frac{1}{9} {R}^{12}_2 {R}^{23}_3-\frac{2}{9} {Q}^{123} {R}^{13}_1-{S}^{1} {S}^{2},
\end{multline}
\begin{multline}\label{eq:dQ123/dx1}
 \frac{\partial {Q}^{123}}{\partial x_1} = \frac{2}{3} {R}^{13}_1 {R}^{12}_1-\frac{1}{3} {R}^{23}_3 {R}^{13}_1 +{Q}^{123} {R}^{13}_3-\frac{1}{3} {R}^{23}_2 {R}^{12}_1 \\
 +{Q}^{123} {R}^{12}_2-{Q}^{123} {S}^{1}+\left({R}^{23}_2 -{R}^{13}_1\right) {S}^{2}+\left({R}^{23}_3 -{R}^{12}_1\right) {S}^{3}.
\end{multline}
All 30 derivatives can be determined from the six shown above by permuting the various indices that appear.

Quite remarkably, as was mentioned in Ref.~\cite{kalnins2011laplace}, the integrability conditions for \eqref{eq:dR12_1/dx1}-\eqref{eq:dQ123/dx1} are identically satisfied.
Therefore any 10-tuple of complex numbers $(\mathbf{Q},\mathbf{R},\mathbf{S})\in\mathbb{C}^{10}$ will give rise to a unique maximum-parameter, second-order conformally superintegrable system over flat space.
The task now is to classify the orbits of $(\mathbf{Q},\mathbf{R},\mathbf{S})$ under the action of the conformal group.
We will show below that, under the local action of the conformal group (i.e.\ excluding translation of the regular point), this problem is solved by considering the root structures of a certain $6$th degree polynomial in one variable $p(z)$, up to fractional linear transformation (M\"obius transformation).

The classification of these root structures will be succinctly described by vanishing of certain covariants of the aforementioned sextic $p(z)$.
The non-local problem is then solved by using these covariants to determine polynomial ideals corresponding to algebraic sets (with some Zariski closed subsets removed) on which the conformal group acts transitively.

\section{Action of the Conformal Group}
Consider the effect of the conformal group on the potential.
Explicitly, this means a change of coordinates for which the Hamiltonian becomes
\[
  H=p_{x_1}^2+p_{x_2}^2+p_{x_3}^2+V(\mathbf{x}) =\frac{p_{u_1}^2+p_{u_2}^2+p_{u_3}^2}{\mu(\mathbf{u})}+V(\mathbf{u})
\]
where $ds^2= \mu(\mathbf{u}) \left({du_1}^2+{du_2}^2+{du_3}^2\right)$ is the flat-space metric expressed in our new coordinate system.
Conformal scaling by $\mu(\mathbf{u})$, as per \eqref{eq:conformal_scaling_of_the_hamiltonian}, gives the flat-space conformally-superintegrable Hamiltonian
\begin{align*}
 \widetilde{H}=H_0+ \mu V .
\end{align*}
If we consider a fixed regular point within our system (i.e.\ excluding translation of the regular point) there are three essential types of non-trivial transformations: Rotation, dilations and M\"obius transformations, i.e.\ conjugation of an inversion in the sphere with a translation. 

\subsection{Inversion in the sphere}\label{sec:inversion_in_the_sphere}

A notable discrete conformal transformation is given by inversion in the unit sphere.
That is, a change of variables of the form 
\begin{align}
	x_i = \frac{u_i}{{u_1}^2+{u_2}^2+{u_3}^2}.\label{transformation_inversion_in_the_sphere}
\end{align}
Under this change of variables the Hamiltonian becomes
\[
 H=\left({u_1}^2+{u_2}^2+{u_3}^2\right)^2\left(p_{u_1}^2+p_{u_2}^2+p_{u_3}^2\right)+V(u_1,u_2,u_3)
\]
which, conformally scaling by $(\mathbf{u}.\mathbf{u})^{-2}$, gives the conformally equivalent Hamiltonian
\[
 \widetilde{H}=\dfrac{H}{\left({u_1}^2+{u_2}^2+{u_3}^2\right)^2}=\left(p_{u_1}^2+p_{u_2}^2+p_{u_3}^2\right)+\frac{V(u_1,u_2,u_3)}{\left({u_1}^2+{u_2}^2+{u_3}^2\right)^2}.
\]
For this conformally-superintegrable Hamiltonian $\widetilde{H}$ we can derive the 10 coefficient functions \eqref{QRS_variable_names} in terms of the original coefficient functions.
Namely, under the action of \eqref{transformation_inversion_in_the_sphere}, these become
\begin{multline}
\widetilde{S}^{1} = -\dfrac{u_{1}^2-u_{2}^2-u_{3}^2}{\left(u_{1}^2+u_{2}^2+u_{3}^2\right)^2} S^{1}
	 -\dfrac{2 u_{1} u_{2}}{\left(u_{1}^2+u_{2}^2+u_{3}^2\right)^2} S^{2}
	 -\dfrac{2 u_{3} u_{1}}{\left(u_{1}^2+u_{2}^2+u_{3}^2\right)^2} S^{3}
	 +\dfrac{2 u_{1}}{\left(u_{1}^2+u_{2}^2+u_{3}^2\right)},
	\label{S1_inversion_in_the_sphere}
\end{multline}
\begin{multline}
\widetilde{R}^{12}_{1} = 
	  \dfrac{2 u_{2} u_{1} \left(6 u_{3}^2 u_{2}^2-2 u_{2}^2 u_{1}^2+u_{2}^4+u_{1}^4-10 u_{3}^2 u_{1}^2+5 u_{3}^4\right)}{\left(u_{1}^2+u_{2}^2+u_{3}^2\right)^4} R^{13}_{3} \\
	+\dfrac{2 u_{2} u_{1} \left(2 u_{3}^2 u_{2}^2-u_{3}^4-2 u_{3}^2 u_{1}^2-10 u_{2}^2 u_{1}^2+3 u_{1}^4+3 u_{2}^4\right)}{\left(u_{1}^2+u_{2}^2+u_{3}^2\right)^4} R^{12}_{2} \\
	+\dfrac{8 u_{3} u_{1}^2 u_{2} \left(2 u_{3}^2-2 u_{2}^2+u_{1}^2\right)}{\left(u_{1}^2+u_{2}^2+u_{3}^2\right)^4} R^{23}_{2} \\
	+\dfrac{\left(-u_{2}^6-15 u_{1}^4 u_{2}^2+15 u_{1}^2 u_{2}^4+u_{1}^6+u_{3}^6-u_{2}^4 u_{3}^2+u_{2}^2 u_{3}^4-u_{1}^2 u_{3}^4+6 u_{1}^2 u_{3}^2 u_{2}^2-u_{1}^4 u_{3}^2\right)}{\left(u_{1}^2+u_{2}^2+u_{3}^2\right)^4} R^{12}_{1} \\
	-\dfrac{2 u_{3} u_{2} \left(u_{3}^4-10 u_{3}^2 u_{1}^2+2 u_{3}^2 u_{2}^2-6 u_{2}^2 u_{1}^2+5 u_{1}^4+u_{2}^4\right)}{\left(u_{1}^2+u_{2}^2+u_{3}^2\right)^4} R^{13}_{1} \\
	-\dfrac{4 u_{1}^2 \left(6 u_{3}^2 u_{2}^2-u_{3}^2 u_{1}^2-u_{2}^4-u_{3}^4+u_{2}^2 u_{1}^2\right)}{\left(u_{1}^2+u_{2}^2+u_{3}^2\right)^4} R^{23}_{3} \\
	+\dfrac{4 u_{3} u_{1} \left(3 u_{2}^4+2 u_{3}^2 u_{2}^2-8 u_{2}^2 u_{1}^2+u_{1}^4-u_{3}^4\right)}{\left(u_{1}^2+u_{2}^2+u_{3}^2\right)^4} Q^{123},
	\label{R12_1_inversion_in_the_sphere}
\end{multline}
\begin{multline}
\widetilde{Q}^{123} = 
	 2 u_{3} u_{1}\dfrac{\left(5 u_{2}^4-u_{3}^4-10 u_{1}^2 u_{2}^2+u_{1}^4\right)}{\left(u_{1}^2+u_{2}^2+u_{3}^2\right)^4} R^{12}_{1}
	+2 u_{3} u_{2}\dfrac{\left(5 u_{1}^4-u_{3}^4-10 u_{2}^2 u_{1}^2+u_{2}^4\right)}{\left(u_{1}^2+u_{2}^2+u_{3}^2\right)^4} R^{12}_{2} \\
	{}
	+2 u_{1} u_{2}\dfrac{\left(5 u_{3}^4-u_{1}^4-10 u_{2}^2 u_{3}^2+u_{2}^4\right)}{\left(u_{1}^2+u_{2}^2+u_{3}^2\right)^4} R^{23}_{2}
	+2 u_{1} u_{3}\dfrac{\left(5 u_{2}^4-u_{1}^4-10 u_{3}^2 u_{2}^2+u_{3}^4\right)}{\left(u_{1}^2+u_{2}^2+u_{3}^2\right)^4} R^{23}_{3} \\
	{} 
	+2 u_{2} u_{3}\dfrac{\left(5 u_{1}^4-u_{2}^4-10 u_{3}^2 u_{1}^2+u_{3}^4\right)}{\left(u_{1}^2+u_{2}^2+u_{3}^2\right)^4} R^{13}_{3}
	+2 u_{2} u_{1}\dfrac{\left(5 u_{3}^4-u_{2}^4-10 u_{1}^2 u_{3}^2+u_{1}^4\right)}{\left(u_{1}^2+u_{2}^2+u_{3}^2\right)^4} R^{13}_{1} \\
	{}
	-\dfrac{\left(-5 u_{2}^2 u_{3}^4-5 u_{1}^2 u_{2}^4+30 u_{1}^2 u_{3}^2 u_{2}^2-5 u_{2}^4 u_{3}^2-5 u_{1}^2 u_{3}^4-5 u_{1}^4 u_{2}^2-5 u_{1}^4 u_{3}^2+u_{1}^6+u_{2}^6+u_{3}^6\right)}{\left(u_{1}^2+u_{2}^2+u_{3}^2\right)^4} Q^{123},
	\label{Q123_inversion_in_the_sphere}
\end{multline}
where, like before, all 10 functions can be determined from the three above by index permutation.

Equation \eqref{S1_inversion_in_the_sphere} is especially important as we will now use it show the conformal group acts transitively on the $S^i$'s.
\begin{theorem}\label{thm:mapping_S1S2S3_to_zero_by_a_local_conformal_transformation}
  Given a regular point in our system there is a local conformal transformation and rescaling that takes the value of the 10 coefficient functions from their original values $(\mathbf{Q}_0,\mathbf{R}_0,\mathbf{S}_0)$ to the values $(\mathbf{Q}_0,\mathbf{R}_0,\mathbf{0})$.
  That is every systems is conformally equivalent to one with $\mathbf{S}_0=0$ at the regular point.
\end{theorem}

\begin{proof}
  Assume $S^3\neq 0$ at our regular point.
  If we perform and inversion in the sphere via \eqref{transformation_inversion_in_the_sphere} such that the values of the transformed regular point satisfy $u_1=u_2=0,u_3\neq0$ then \eqref{S1_inversion_in_the_sphere}-\eqref{Q123_inversion_in_the_sphere} we get
  \begin{align}
    \widetilde{S}^1= \frac{S^1}{u_3^2},\quad
    \widetilde{S}^2= \frac{S^2}{u_3^2},\quad
    \widetilde{S}^3=-\frac{S^3}{u_3^2}+\frac{2}{u_3}.
  \end{align}
  Dilating by a factor of $\delta=u_3^2$ we find that the $\widehat{Q}^{123},\widehat{R}^{ij}_i$ have returned to their orginal value while the $S^i$'s have become
  \begin{align}
    \widetilde{S}^1=S^1,\quad
    \widetilde{S}^2=S^2,\quad
    \widetilde{S}^3=-S^3+2u_3.
  \end{align}
  Making the choice $u_3=\frac{S^3}{2}$ we now have have $\widetilde{S}^3=0$.
  Permuting the indexes allows the same to be done for $S_1$ and $S_2$, and finally undoing these index permutations returns $\mathbf{Q},\mathbf{R}$ to their original values $\mathbf{Q}_0,\mathbf{R}_0$.
\end{proof}

\subsection{Rotation adapted variables}\label{sec:rotation_apapted_variables}

In the previous section we were able to successfully deal with the variable sets $\{\mathbf{Q},\mathbf{R}\}$ and $\{\mathbf{S}\}$ independently of each other.
One of the reasons this is possible is because the two pairs of variable form distinct rotation representations.
To make this more explicit we will consider the action of the $\mathfrak{so}(3,\mathbb{C})$ Lie algebra on these variables.
Define $J_1,J_2$ and $J_3$ to be the actions of the $\mathfrak{so}(3,\mathbb{C})$ Lie algebra corresponding to an infinitesimal rotation around the $x_1,x_2$ and $x_3$ axes respectively, These satisfy the commutation relations
\[
        \left[J_1,J_2\right]=J_3,
  \quad \left[J_2,J_3\right]=J_1,
  \quad \left[J_3,J_1\right]=J_2,
\]
where $\left[\cdot,\cdot\right]$ is the commutator $\left[x,y\right]=xy-yx$.
Using the operators above we can define the standard set of raising and lowering operators via
\begin{align}
  J_{+}=i J_{1}-J_{2},\quad 
  J_{0}=i J_{3},      \quad
  J_{-}=i J_{1}+J_{2}.\label{raising_lowering_J_operators}
\end{align}
These satisfy the commutation relations
\[
        \left[J_{0},J_{+}\right] =  J_{+},
  \quad \left[J_{+},J_{-}\right] =2 J_{0},
  \quad \left[J_{0},J_{-}\right] = -J_{-}.
\]
As is well known, the operators $J_{+}$ and $J_{-}$ map between the eigenspaces of the $J_{0}$ operator by respectively raising and lowering eigenvalues by $\pm1$.
Using \eqref{raising_lowering_J_operators} we can define a three and a seven-dimensional representation in the span of the functions $\mathbf{Q},\mathbf{R},\mathbf{S}$.
These representations, which are invariant subspaces under the Lie algebra action, are characterised by highest weight eigenvectors with respective eigenvalue $l=+1$ and $l=+3$.
The basis vectors of this representations will be normalised such that the action of \eqref{raising_lowering_J_operators} on an eigenvector $f_m$ with eigenvalue $m$ in a representation with highest-weight $l$ is given by
\begin{align}
  J_{+} f_m&=\sqrt{(l-m)(l+m+1)}f_{m+1}, \nonumber \\
  J_{0} f_m&=m f_m,                      \nonumber \\
  J_{-} f_m&=\sqrt{(l+m)(l-m+1)}f_{m-1}. \label{eigenvector_normalization}
\end{align}
Using these representation we define the rotation adapted variables
\begin{align}
X_{+1} &= i S_2+S_1,     \nonumber \\
X_{ 0} &= -S_3 \sqrt{2}, \nonumber \\
X_{-1} &= i S_2-S_1,     \label{eq:the_X_representation}
\end{align}
and
\begin{align}
Y_{+3} &= R^{12}_{1}+\frac{1}{4} R^{23}_{3}+i\left( R^{12}_{2}+\frac{1}{4} R^{13}_{3}\right), \nonumber\\
Y_{+2} &= \frac{1}{4}\sqrt{6} \left(i\left(R^{13}_{1}-R^{23}_{2}\right)-2 Q^{123}\right),     \nonumber\\
Y_{+1} &= \frac{1}{4}\sqrt{15} \left(R^{23}_{3}-i R^{13}_{3}\right),                          \nonumber\\
Y_{0}  &= -\frac{1}{2} i \sqrt{5} \left(R^{13}_{1}+R^{23}_{2}\right),                         \nonumber\\
Y_{-1} &= \frac{1}{4}\sqrt{15} \left(R^{23}_{3}+i R^{13}_{3}\right),                          \nonumber\\
Y_{-2} &= \frac{1}{4}\sqrt{6} \left(i\left(R^{13}_{1}-R^{23}_{2}\right)+2 Q^{123}\right),     \nonumber\\
Y_{-3} &= R^{12}_{1}+\frac{1}{4} R^{23}_{3}-i\left( R^{12}_{2}+\frac{1}{4} R^{13}_{3}\right). \label{eq:the_Y_representation}
\end{align}

The partial derivatives can also be written up in the manner of a representation via
 \begin{align}
	\partial_{+}&=i\partial_{x_2}+\partial_{x_1},\nonumber\\
	\partial_{0}&=\partial_{x_3}              ,\nonumber\\
	\partial_{-}&=i\partial_{x_2}-\partial_{x_1},\label{eq:raising_lowering_and_level_set_derivatives}
\end{align}
which has been normalised to satisfy the commutation relations given in table~\ref{table:the_J_Diff_commutation_relations}.
\begin{table}
\begin{center}
 \begin{tabular}{ c | c c c }
   $\left[ J_{\alpha},\partial_{\beta} \right]$  & $\partial_{+}$ & $\partial_{0}$ &  $\partial_{-}$ \\
               \hline  $J_{+}$ &              $0$ & $-\partial_{+}$ & $-2\partial_{0}$ \\
                       $J_{0}$ &   $\partial_{+}$ &             $0$ &  $-\partial_{-}$ \\
                       $J_{-}$ & $-2\partial_{0}$ & $-\partial_{-}$ &              $0$ \\
\end{tabular}
\end{center}
\caption{The $J$, $\partial$ commutation relations}\label{table:the_J_Diff_commutation_relations}
\end{table}

Applying $J_0$ operator to $\partial_+$, $\partial_0$, and $\partial_-$ and using the commutation relations in table~\ref{table:the_J_Diff_commutation_relations}, it is easy to show that if $F(\mathbf{Y})$ is an eigenvector of $J_0$ with eigenvalue $\lambda$ then $\partial_+(F(\mathbf{Y}))$, $\partial_0(F(\mathbf{Y}))$, $\partial_-(F(\mathbf{Y}))$ are also eigenvectors of $J_0$ with respective eigenvalues $\lambda+1$, $\lambda$, and $\lambda-1$. 

\begin{theorem}\label{thm:showing_the_derivative_of_a_non_highest_weight_vector_is_the_lowering_of_derivatives_of_highest_weight_vectors}
The derivatives of a polynomial in the $Y_i$'s can always be written in terms of the lowering operator $J_-$ acting on derivatives of highest weight vectors. 
\end{theorem}
\begin{proof}
If we rewrite the commutation relations in table~\ref{table:the_J_Diff_commutation_relations} we have, for any polynomial function $F(\mathbf{Y})$, 
\begin{align}
  \partial_{+}J_{-}\left(F(\mathbf{Y})\right) &= J_{-}\partial_{+}\left(F(\mathbf{Y})\right)+ 2 \partial_{0}\left(F(\mathbf{Y})\right), \nonumber\\
  \partial_{0}J_{-}\left(F(\mathbf{Y})\right) &= J_{-}\partial_{0}\left(F(\mathbf{Y})\right)+   \partial_{-}\left(F(\mathbf{Y})\right), \nonumber\\
  \partial_{-}J_{-}\left(F(\mathbf{Y})\right) &= J_{-}\partial_{-}\left(F(\mathbf{Y})\right).
\end{align}
So given the derivative of something which isn't a highest weight vector, i.e.\ powers of the lowering operator acting on highest-weight vectors, the relations above can be used to commute the lowering operator $J_-$ with the derivatives until the derivatives act only upon highest-weight vectors.  
\end{proof}

A consequence of this theorem is that there is no information lost by just focusing on the derivatives of the highest weight vector of a representation.
An alternatively way to understand this fact is that the derivatives of the elements an ideal, closed under the action of rotations and dilations, are 
themselves closed ideals, and these can be determined from their highest weight vectors.
As such only the derivatives of the highest weight vectors need be examined when determining the differential closure of the ideals in section~\ref{sec:ideals_covariants_and_root_structures}.

The partial derivatives of a highest weight vector are not themselves highest weight vector.
As highest weight vectors are our stand-in for the full representations, it would be useful to be know how to construct the highest weight vectors in the derivatives directly from a given highest-weight vector.

\begin{theorem}\label{thm:constructing_highest_weight_vectors_from_the_derivatives_of_the_highest_weight_vector}
  Given a representation with normalisation of the form \eqref{eigenvector_normalization}, where $l>1$, there are three representations made from their partial derivatives and they have the highest-weight vectors
  \begin{align}
    g_{l+1}&= \partial_{+}\left(f_l\right),                                                                                                            \label{eq:the_naive_constructed_raising_derivative}\\
    g_{l  }&= \partial_{0}\left(f_l\right)+\frac{1}{2l}\partial_{+}\left(J_{-}(f_{l})\right),                                                          \label{eq:the_naive_constructed_level_set_derivative}\\
    g_{l-1}&= \partial_{-}\left(f_l\right)+\frac{1}{l}\partial_{0}\left(J_{-}(f_{l})\right)+\frac{1}{2l(2l-1)}\partial_{+}\left(J_{-}^2(f_{l})\right). \label{eq:the_naive_constructed_lowering_derivative}
  \end{align}
  If instead $l=0$ then $g_{1} = \partial_{+}\left(f_0\right)$ suffices.
\end{theorem}

\begin{proof} Using the commutation relations in table~\ref{table:the_J_Diff_commutation_relations} and the action 
  \[
    J_{+}(f_l)=0,\quad J_{0}(f_l)=lf_l
  \]
  we can apply $J_{+}$ to \eqref{eq:the_naive_constructed_raising_derivative} to show
  \begin{align}
    J_{+}\left(g_{l+1}\right)
      &= J_{+}\partial_{+}\left(f_l\right) \nonumber\\
      &= \partial_{+}J_{+}\left(f_l\right) \nonumber\\
      &= 0.
  \end{align}
  Similiary, for $l>1$, applying $J_{+}$ to \eqref{eq:the_naive_constructed_level_set_derivative} shows
  \begin{align}
    J_{+}\left(g_{l}\right)
      &= J_{+} \partial_{0}\left(f_l\right)+\frac{1}{2l}J_{+}\partial_{+}\left(J_{-}(f_l)\right) \nonumber\\
      &= \bigg(\partial_{0}J_{+}\left(f_l\right)-\partial_{+}\left(f_l\right)\bigg)
         +\frac{1}{2l}\partial_{+}J_{+}J_{-}(f_l) \nonumber\\
      &= -\partial_{+}\left(f_l\right)
         +\frac{1}{2l}\partial_{+}\bigg(J_{-}J_{+}(f_l)+2J_{0}(f_l)\bigg) \nonumber\\
      &= -\partial_{+}\left(f_l\right)
         +\partial_{+}\left(f_l\right) \nonumber\\
      &= 0. 
  \end{align}
  And finally, for $l>1$, applying $J_{+}$ to \eqref{eq:the_naive_constructed_lowering_derivative} shows
  \begin{align}
    J_{+}\left(g_{l-1}\right)
      &= J_{+}\partial_{-}\left(f_l\right)
         +\frac{1}{l}J_{+}  \partial_{0}J_{-}(f_l)
         +\frac{1}{2l(2l-1)}J_{+}\partial_{+}J_{-}^2(f_l) \nonumber\\
      &= \bigg(\partial_{-}J_{+}\left(f_l\right)-2\partial_{0}\left(f_l\right)\bigg)
         +\frac{1}{l}\bigg(\partial_0J_{+}J_{-}(f_l)-\partial_{+}J_{-}(f_l)\bigg)\nonumber\\
      &\phantom{+{}} +\frac{1}{2l(2l-1)}\partial_{+}J_{+}J_{-}^2(f_l) \nonumber\\
      &= 2\partial_{0}\left(f_l\right)
         +\frac{1}{l}\partial_0\bigg(J_{-}J_{+}(f_l)+2J_{0}(f_l)\bigg)\nonumber\\
      &\phantom{+{}} -\frac{1}{l}\partial_{+}J_{-}(f_l)
         +\frac{1}{2l(2l-1)}\partial_{+}\bigg(J_{-}^2J_{+}(f_l)+4J_{-}J_{0}(f_l)-2J_{-}(f_l)\bigg) \nonumber\\
      &= -2\partial_0\left(f_l\right)
         +2\partial_0\left(f_l\right)
         +\frac{1}{l}\partial_{+}J_{-}(f_l)
         -\frac{1}{l}\partial_{+}J_{-}(f_l),\nonumber \\
      &= 0.
  \end{align}
  This proves that \eqref{eq:the_naive_constructed_raising_derivative},\eqref{eq:the_naive_constructed_level_set_derivative} and \eqref{eq:the_naive_constructed_lowering_derivative} vanish under $J_+$.
  Since these are eigenvectors of the $J_0$ operator (with three distinct eigenvalues) these form highest-weight vectors for three distinct representations.
  The three representations are purely first-order in terms of the partial derivatives of the $f_i$ and a simple count of dimensions convinces us these representations are sufficient to cover the space of first derivatives.
\end{proof}

\subsection{Decomposition of the Derivatives}\label{sec:decomposition_of_the_derivatives}

In section~\ref{sec:ideals_covariants_and_root_structures} and section~\ref{sec:the_full_classification} we will be considering polynomial ideals formed solely from the $Y_i$'s.
We've already seen, by theorem~\ref{thm:mapping_S1S2S3_to_zero_by_a_local_conformal_transformation}, that the value of the three-dimensional representation is unimportant when classifying systems up to a conformal transformation.
This is perhaps surprising as the $X_i$'s explicitly appear in the derivatives of the $Y_i$'s and one might expect to be able to use these terms to escape from an algebraic variety defined solely in terms of the $Y_i$'s.

Returning, for a moment, to work with the variables $S^i$ instead of $X_i$, an examination of \eqref{eq:dR12_1/dx1}-\eqref{eq:dQ123/dx1} reveals that the partial derivatives an element of the 7-dimensional representation can be written in the form
\begin{align}
  \partial_i(Y)=\widehat{\partial_i}(Y)-S^i D(Y) +J_1(S^i)J_1(Y)+J_2(S^i)J_2(Y)+J_3(S^i)J_3(Y) \label{eq:breaking_up_the_action_of_the_derivative_on_Y}
\end{align}
for $i\in\{1,2,3\}$.
Here $D$ and $J_i$ are the Lie algebra action of the dilations and the rotations, and  $\widehat{\partial_i}$ are the terms in \eqref{eq:dR12_1/dx1}-\eqref{eq:dQ123/dx1} that do not involve the three-dimensional representation  (i.e.\ the $S^i$'s).
Importantly, the form of \eqref{eq:breaking_up_the_action_of_the_derivative_on_Y} shows that the operator
\[
 \widehat{\partial}_i(Y)=\left.\partial_i(Y)\right|_{\mathbf{S}=0}
\]
is a derivation.
As such the operator $\widehat{\partial}_i$ satisfies a Leibniz rule, ensuring that \eqref{eq:breaking_up_the_action_of_the_derivative_on_Y} applies to polynomial combinations of the $Y_i$'s as well.
This shows that the terms involving the $S^i$ can be dropped from our computations when considering polynomial ideal already closed under rotations and dilations. 

Since the values of the $X_i$'s are irrelevant and since only the partial derivatives of the highest weight vector $Y_{+3}$ need be known, the information relevant to our classification can be recovered completely from the following three equations
\begin{align}
 \widehat{\partial}_{+}(Y_{+3}) &=  \frac{2i}{9}\left(Y_{+2}\right)^2-\frac{4i}{45\sqrt{15}} Y_{+1} Y_{+3}, \nonumber\\
 \widehat{\partial}_{0}(Y_{+3}) &=  \frac{i\sqrt{5}}{15} Y_{+3} Y_{ 0} -\frac{i\sqrt{10}}{45} Y_{+1} Y_{+2}, \nonumber\\
 \widehat{\partial}_{-}(Y_{+3}) &= -\frac{2i\sqrt{15}}{5}Y_{-1}Y_{+3} +\frac{17i\sqrt{30}}{45}Y_{+2}Y_{ 0} -\frac{10i}{9}\left(Y_{+1}\right)^2.\label{eq:the_hatted_derivatives} 
\end{align}

\subsection{Reinterpreting the local action of the conformal group}\label{sec:local_action_of_the_conformal_group_on_the_seven_dimensional_representation}

If we use $C^i$ to denote the Lie algebra action of conjugation of an inversion in the sphere \eqref{transformation_inversion_in_the_sphere} with a infinitesimal translation in the $x_i$ direction, then we find that, when restricted to the $Y_i$'s, the Lie algebra action can be written in a form similar to \eqref{eq:breaking_up_the_action_of_the_derivative_on_Y}, namely
\begin{align}
  C_i(Y)=2 x_i D(Y)  + J_1(x_i)J_1(Y) + J_2(x_i)J_2(Y) + J_3(x_i)J_3(Y).\label{action_of_the_mobius_transformation}
\end{align}
So the action of the conformal Lie algebra on the 7-dimensional representation can be decomposed into dilations and rotations.

This combination of dilations and rotations allows the space of representations to be modeled as fixed degree polynomials in one variable.
More precisely the fact that $SL(2,\mathbb{C})$ is the double cover of $SO(3,\mathbb{C})$ allows us to think of our $SO(3,\mathbb{C})$ representations as odd-dimensional $SL(2,\mathbb{C})$ representations.
Given a $(2n+1)$-dimensional $SO(3,\mathbb{C})$ representation, an equivalent action by $SL(2,\mathbb{C})$ can be given be given by an $(2n)$th degree polynomial acted on by a $2 \times 2$ complex matrix in the following manner. 

Consider our seven dimensional representation \eqref{eq:the_Y_representation}, from it we can define the sextic polynomial 
\begin{align}
 p(z)&= \sum_{n=0}^6 (-1)^m\sqrt{\binom{6}{m}}Y_{m-3} z^n \nonumber \\
     &= Y_{-3}-\sqrt{6}Y_{-2}z+\sqrt{15}Y_{-1} -\ldots+Y_{+3}z^6. \label{the_polynomial_p_z}
\end{align}
The action of a matrix $\left(\begin{smallmatrix}a & c\\ b & d\\ \end{smallmatrix}\right)\in SL(2,\mathbb{C})$ is given via the fractional linear transformation
\begin{align}
 p(z)\mapsto(cz+d)^{n}p\left(\dfrac{az+b}{cz+d}\right).\label{action_of_a_matrix_on_P}
\end{align} The action of a dilation can also be represented in the form \eqref{action_of_a_matrix_on_P} through the matrix $\left(\begin{smallmatrix}\delta & 0\\ 0 & \delta\\ \end{smallmatrix}\right)$ and taken together these clearly give the action of $GL(2,\mathbb{C})$.

If we consider the conformal transformation given by the inversion in the sphere \eqref{transformation_inversion_in_the_sphere}, the seemingly complicated action of \eqref{R12_1_inversion_in_the_sphere}-\eqref{Q123_inversion_in_the_sphere} can be succinctly encoded by the action of \eqref{action_of_a_matrix_on_P} under the matrix
\begin{align}
\begin{pmatrix} a & c\\ b & d\\ \end{pmatrix}
    =(u_1^2+u_2^2+u_3^2)^{(2/3)}\begin{pmatrix}u_3 & -u_1+iu_2\\ -u_1-iu_2 & -u_3\\ \end{pmatrix}.
    \label{eq:the_matrix_for_the_action_of_the_inversion_on_the_sphere_for_Y}
\end{align}

So we can easily understand the classification for a fixed regular point, it is the classification of a sextic under the action \eqref{action_of_a_matrix_on_P}, or alternatively, the classification of 6 points in $\mathbb{C}^*$ under the action of $\mbox{PGL}(2,\mathbb{C})$.

\section{Differentially Closed Polynomials Ideals, Covariants and Root Structure}\label{sec:ideals_covariants_and_root_structures}

The canonical forms of $p(z)$ under the action of $GL(2,\mathbb{C})$ can be given in terms of absolute invariants, in this case,
root multiplicities and multi-ratios.
However this information is difficult to handle algebraically when considering translation of the regular point.
It would be preferable to encode the information in terms of the $Y_i$'s in an invariant manner.
A natural way to do this is to consider rotationally closed polynomial ideals generated by homogeneous polynomials in terms of the $Y_i$'s.
This description is easier to deal with algebraically, but even a simple form, such as a sextic with a root of multiplicity four and a root of multiplicity two, is described by a sizable set of  generators.
Thankfully, since we think of our representation as a polynomial in one variable acted on by $GL(2,\mathbb{C})$, we are able to express our representations in terms of covariants of $p(z)$, for which a Hilbert basis can be constructed.
Written in terms of this Hilbert basis the ideals can be described in a few lines.

\subsection{Polynomial Covariants}

Consider a $j$th degree polynomial $p(z;\mathbf{a})$, where the components of $\mathbf{a}$ are the coefficients of $p$.
We induce an action of $GL(2,\mathbb{C})$ on the coefficients $\mathbf{a}$ by making the identification 
\[
 \left(cz+d\right)^j p\left(\frac{az+b}{cz+d};\mathbf{a}\right)=p(z;\widetilde{\mathbf{a}}).
\]
A covariant is a $k$th degree polynomial $q(z;\mathbf{a})$, whose coefficients are functions of the original coefficients $\mathbf{a}$, such that the following holds
\begin{align}
 \left(cz+d\right)^k q\left(\frac{az+b}{cz+d};\mathbf{a}\right)=\Delta^m q (z;\widetilde{\mathbf{a}}),\label{eq:the_definition_of_a_covariant}
\end{align}
where $\Delta=\det\left(\begin{smallmatrix}a & c\\ b & d\\ \end{smallmatrix}\right)$. The exponent $m$, is called the covariant-weight.
By the identification above, $p(z;\mathbf{a})$ has covariant-weight zero.

\subsection{A Hilbert Basis for covariants of $p(z)$}

The Hilbert basis for the ring of covariants is, simply put, a set of covariants such that every other covariant can be written as a (not necessarily unique) polynomial in the given basis covariants.
Gordon's method (describe in Ref.~\cite{MR1694364}) provides an algorithm for computing such a Hilbert basis.
The main tool used in the constructing of the Hilbert basis is the so-called transvectant\footnote{strictly speaking the algorithm uses the `partial transvectant' but the result obtained can always be rewritten in terms of the full transvectant.}.
The $r$th transvectant of an $m$th degree polynomial $Q$ and an $n$th degree polynomial $R$ is given by
\begin{align*}
 \left(Q,R\right)^{(r)}=r!\sum_{k=0}^{r}(-1)^k\binom{n-r+k}{k}\binom{m-k}{r-k}Q^{(r-k)}(z)R^{(k)}(z)
\end{align*}
where $Q^{(k)}(z)$ is the $k$th derivative of $Q$ with respect to $z$.

The Hilbert basis for our sextic is composed of 26 covariants.
These covariant are given in table~\ref{table:the_hilbert_basis} and are assigned names consisting of a capital letter with a numerical subscript.
The position of the letter in the English alphabet indicates the polynomial degree of the coefficients in terms of the $Y_i$'s and the subscript indicates the weight of the highest-weight vector.
This convention, with the addition of descriptive superscripts, will be used when describing the covariants whose coefficients generate the polynomials ideals in section~\ref{sec:the_full_classification}.
\renewcommand{\arraystretch}{1.5}
\begin{table}
\begin{center}
 \begin{tabular}{|c |  l l|}
 \hline
Order & Covariants & \\
\hline
1
&  $A_3 = p(z)$ 
& \\
\hline
2
& $B_4 = (A_3,A_3)^{(2)}$
& $B_2 = (A_3,A_3)^{(4)}$ \\
& $B_0 = (A_3,A_3)^{(6)}$ 
& \\
\hline
3
& $C_6 = \frac{1}{2}(A_3,B_4)^{(1)}$
& $C_4 = \frac{1}{2}(A_3,B_2)^{(1)}$ \\
& $C_3 = \frac{1}{6}(A_3,B_2)^{(2)}+5A_3 B_0$
& $C_1 = \frac{1}{6}(A_3,B_2)^{(4)}$ \\
\hline
4
& $D_5 = \frac{2}{3}(A_3,C_3)^{(1)}$
& $D_3 = (A_3,C_1)^{(1)}$ \\
& $D_2 = (A_3,C_1)^{(2)}$
& $D_0 = \frac{1}{15}(A_3,C_3)^{(6)}$ \\
\hline
5
& $E_4 = (A_3,D_2)^{(1)}$
& $E_2 = (A_3,D_2)^{(3)}$ \\
& $E_1 = (A_3,D_2)^{(4)}$ 
& \\
\hline
6
& $F_3^{(1)} = (A_3,E_1)^{(1)}$
& $F_0       = (C_1,C_1)^{(2)}$ \\
& $F_3^{(2)} = -\frac{5}{2}F_3^{(1)}+\frac{1}{2}(A_3,E_2)^{(2)}$
& \\
\hline
7
& $G_2 = \frac{1}{2}(B_2,E_1)^{(1)}$
& $G_1 = (C_1,D_2)^{(2)}$ \\
\hline
8
& $H_1 = \frac{-1}{12}(A_3,G_2)^{(4)}$ 
& \\
\hline
9
& $I_2 = (A_3,H_1)^{(2)}$ 
& \\
\hline
10
& $J_1 = (A_3,I_2)^{(4)}$
& $J_0 = (C_1,G_1)^{(2)}$ \\
\hline
12
& $L_1 = \frac{1}{6}(B_2,J_1)^{(2)}+\frac{1}{3}B_0 J_1$ 
& \\
\hline
15
& $O_0 = (L_1,C_1)^{(2)}$ 
& \\
\hline
\end{tabular}
\end{center}
\caption{A Hilbert basis for the covariants of $p(z)$}\label{table:the_hilbert_basis}
\end{table}
\renewcommand{\arraystretch}{1.0}

\subsection{Two simple ideals}\label{sec:Two_Relative_Invariants}

Before discussing the full classification result it is worth going over two illustrative examples.
The first is just to note that if the all the $Y_i$'s that carry the seven dimensional representation \eqref{eq:the_Y_representation} take the value zero, then it is clear from \eqref{eq:the_hatted_derivatives} that the value of the derivatives of the $Y_i$'s are, to any order, zero as well.
Hence a system will satisfy the conditions $Y_{+3},\ldots,Y_{-3}=0$ at one point if and only if it does so everywhere, as will any conformal scaling of this system.
This immediately allows us to split the systems into two distinct conformal classes: Those with coefficient functions $Y_i$ that lie in the polynomial ideal generate by the covariant $A_3$, and those that don't.

A more complicated example is given by considering the Hessian of our sextic.
For any non-zero polynomial the vanishing of the Hessian implies the polynomial has a single root.
For a $6{th}$ order polynomial, the Hessian is the $8{th}$ order covariant given by
\begin{align}
  H(z) 
    &= 30\left(p^{\prime\prime}(z)p(z)-\frac{6}{5}\left(p^\prime(z)\right)^2\right) \nonumber \\
    &= 150 (\frac{6}{\sqrt{15}} Y_{+3} Y_{+1}-Y_{+2}^2) z^8+\ldots+150 (\frac{6}{\sqrt{15}} Y_{-3} Y_{-1}-Y_{-2}^2) \nonumber\\
    &=\frac{1}{2}\left(B^{(4)}_{+4} z^8-\sqrt{8}B^{(4)}_{+3} z^7+\ldots+B^{(4)}_{-4}\right)\nonumber\\
    &= \frac{1}{2}B_4, \label{the_hessian} 
\end{align}
where $B_4$ is taken from the Hilbert basis in table~\ref{table:the_hilbert_basis}.
For future reference the coefficients of $B_4$ have been given the names $B^{(4)}_{i}$, $i\in\{-4,\ldots,+4\}$.
If we consider the action of the restricted derivatives \eqref{eq:the_hatted_derivatives} on the highest-weight term in \eqref{the_hessian} we find the following relations
\begin{align}
 \widehat{\partial}_{+}\left(B^{(4)}_{+4}\right)
       &= \left(-6X_{+1}-\frac{28}{9\sqrt{15}}iY_{+1}\right)B^{(4)}_{+4}
         +\frac{14}{9\sqrt{3}}iY_{+2}B^{(4)}_{+3}-\frac{2\sqrt{7}}{9}iY_{+3}B^{(4)}_{+2},
\end{align}
\begin{align}
 \widehat{\partial}_{0}\left(B^{(4)}_{+4}\right)
       &= \left(\sqrt{2}X_{ 0}+\frac{56i}{9\sqrt{5}} Y_{0}\right)B^{(4)}_{+4}
         +\left(\sqrt{2}X_{+1}-\frac{77i\sqrt{2}}{9\sqrt{15}}Y_{+1}\right)B^{(4)}_{+3} \nonumber\\
       &\phantom{=} +\frac{10i\sqrt{14}}{9\sqrt{3}} Y_{+2}B^{(4)}_{+2}
         -\frac{i\sqrt{14}}{3}B^{(4)}_{+1}Y_{+3},
\end{align}
So the derivatives  $\partial_{+}\left(B^{(4)}_{+4}\right)$ and $ \partial_{0}\left(B^{(4)}_{+4}\right)$ are contained in the ideal formed by the coefficients of the Hessian.
From theorems~\ref{thm:showing_the_derivative_of_a_non_highest_weight_vector_is_the_lowering_of_derivatives_of_highest_weight_vectors} and~\ref{thm:constructing_highest_weight_vectors_from_the_derivatives_of_the_highest_weight_vector} this implies that the 11-dimensional and 9-dimensional representations made from the partial derivatives of the coefficient of the Hessian are contained in this ideal.

However if we look at the lowering derivative we find
\begin{align}
 \widehat{\partial}_{-}\left(B^{(4)}_{+4}\right)
       &= \left(2X_{-1}+\frac{16i\sqrt{3}}{\sqrt{5}}Y_{-1}\right)B^{(4)}_{+4}
         -\left(2X_{ 0}+\frac{8i\sqrt{2}}{3\sqrt{5}}Y_{ 0}\right)B^{(4)}_{+3} \nonumber\\
       &\phantom{=} -\frac{194i}{3\sqrt{105}}Y_{+1}B^{(4)}_{+2}
         +\frac{14i\sqrt{7}}{3\sqrt{3}}Y_{+2}B^{(4)}_{+1} 
         -\frac{2i\sqrt{14}}{\sqrt{5}}Y_{+3}B^{(4)}_{ 0} \nonumber\\
       &\phantom{=} +\frac{176i}{105\sqrt{15}}\left(Y_{+1}^2-\frac{\sqrt{10}}{\sqrt{3}}Y_{ 0}Y_{+2}+\frac{\sqrt{5}}{\sqrt{3}}Y_{-1}Y_{+3}\right)Y_{+1},\label{eq:lowering_diff_of_B4}
\end{align}
and the final term is not in the ideal formed by the coefficients of Hessian.
Adding these cubic conditions gives a differentially closed ideal.
However a simple check reveals that these extra conditions vanish if our sextic has the multiplicity six root implied by the Hessian vanishing, so they do not actually place any further restrictions on this root structure.

To avoid confusion it would be best to work with radical ideals, i.e.\ ideals containing all polynomials which
vanish on their associated variety.
In the case of the Hessian this would be the ideal generated by the coefficients of the covariants $B_0,B_2$ and $B_4$.

Any ideal can be differentially closed by adding derivatives to the set of generators of the ideal.
This is guaranteed to be a finite process due to the Noetherian property polynomial rings over $\mathbb{C}$.
We denote the output of this differential closure procedure on an ideal $I$ by $\overline{I}$, and we denote the radical by $\sqrt{I}$.
Once a differential closed ideal is found we can, conceptually at least, begin to discuss the radical the ideal.

\begin{lemma}
  If an ideal $I$ is closed under differentiation then so is the radical of the ideal, $\sqrt{I}$.
  That is to say
  \[
    \overline{\sqrt{\overline{J}}}=\sqrt{\overline{J}}
  \]
  for any ideal $J$.
\end{lemma}
\begin{proof}
  If $A$ is in the radical of the ideal $I$ then $A^n \in I$ for some integer $n$.
  Taking the $n$th derivative of $A^n$ with respect to $x_i$ gives
  \begin{align}
    \partial_i^n(A^n)=n!(\partial_iA)^n+O(A),\label{eq:the_nth_partial_derivative_of_A_to_the_n}
  \end{align}
  where the terms hidden in $O(A)$ are at least first order in $A$.
  Rearranging \eqref{eq:the_nth_partial_derivative_of_A_to_the_n} shows that $(\partial_iA)^n$ can be written as a sum of terms in $\sqrt{I}$.
  Since this is a radical ideal, $\partial_iA$ must be in $\sqrt{I}$ as well.
  Hence $\sqrt{I}$ is  differentially closed.
\end{proof}

\subsection{Ideals defining co-incident root-loci}\label{sec:ideals_generated_from_coincident_root_loci}

A starting point for our search for differentially closed polynomial ideals will be given by considering the necessary conditions for different root multiplicities.
If we include the case where the sextic is identically zero then are 12 different multiplicities possible.
These correspond to the 11 partitions of six and the case $p(z)=0$; we will denote these 12 cases by
  $$[111111],[21111],[2211],[222],[3111],[321],[33],[411],[42],[51],[6],[0]$$
where, for example, $[21111]$ denotes a double root, and $[0]$ indicates the zero-polynomial.

While there are interesting theoretical techniques for describing the ring of covariants vanishing under certain root structure (e.g.\ see Ref.~\cite{MR2102640}) it was simpler in this case to simply construct the ideals using elimination ideals in the algebra software Singular.
The use of elimination ideals has the advantage that the ideals obtained are radical by construction.

\section{The Full Classification}\label{sec:the_full_classification}

\paragraph{Case [0]:} The simplest form for the sextic is given by $p(z)\equiv0$, that is, the coefficients $Y_{\pm m}=0$.
In terms of the Hilbert basis (table~\ref{table:the_hilbert_basis}) this, somewhat trivially, is represented by the ideal generate from the coefficients of the covariant $A_3$.
This ideal will be denoted $I_{[0]}$.

As was already discussed in section~\ref{sec:Two_Relative_Invariants}, this ideal is closed under differentiation (i.e.\ $\overline{I_{[0]}}=I_{[0]}$).
Thus, if the elements of the ideal $I_{[0]}$ are simultaneously zero, they remain so under translation of the regular point.
There is only one 7-tuple $Y_{i}\in\mathbb{C}^7$ which satisfies this ideal, that is, only 7-tuple 
which is a common zero of all polynomials in
the ideal.  Hence this ideal classifies a single conformal class.

A representative of the systems in the $[0]$-type class is given by the isotropic oscillator on flat space,
\begin{align}
 V_O= a (x_1^2+x_2^2+x_3^2)+b x_1+c x_2+d x_3+e.
\end{align}
The classifying sextic for $V_0$ is, of course,
\[
 p(z)=0.
\]

\paragraph{Case [6]:}The ideal $I_{[6]}$ containing the conditions for the $[6]$-type root structure is given by the coefficients of the covariants
\begin{align*}
 B_{ 4}^{[6]} &= 
       B_4,  \\
 B_{ 2}^{[6]} &= 
       B_2,  \\
 B_{ 0}^{[6]} &= 
       B_0.  
\end{align*}
The ideal $I_{[6]}$ consists of 15 second-order polynomials and it can easily be shown that the cubic polynomials that arise from the derivatives are contained in $I_{[6]}$, i.e.\ $\overline{I_{[6]}}=I_{[6]}$.
So just like the previous case this shows the vanishing of $I_{[6]}$ locally implies it also vanishes globally, meaning $I_{[6]}$ either vanishes identically, or not at all.
This implies the root structure $[6]$ is persistent feature when found at a regular point in a system (i.e.\ it remains a feature in an open set around that point).
Since $I_{[0]} = 0$ is also a persistent feature, the $[6]$-type sextic cannot degenerate into a $[0]$-type sextic at a non-generic point.
The local action of the conformal group, acting like $GL(2,\mathbb{C})$, is transitive on three or fewer roots, and therefore any two sextics in this class are equivalent.
Hence this corresponds to a single conformal class.

A particular representative of the $[6]$-type systems is the (Euclidean superintegrable) system
\begin{multline}
 V_{A} =  a \left((x_1-i x_2)^3+6({x_1}^2+ {x_2}^2+ {x_3}^2)\right)+b \left((x_1-i x_2)^2+2(x_1+ i x_2)\right) \\
         +c (x_1-i x_2)+d x_3+e.
\end{multline}
The classifying sextic for $V_{A}$ is given by
$$
 p(z)=iz^6,
$$
which has a leading coefficient that clearly vanishes nowhere.

\paragraph{Case [51]:}The ideal $I_{[51]}$ of conditions for the $[51]$ root structure is generated by a subset of the $I_{[6]}$ generators, namely the coefficients of the covariants
\begin{align*}
 B_{ 2}^{[51]} &= 
       B_2,  \\
 B_{ 0}^{[51]} &= 
       B_0.  
\end{align*}
As before it is simple to show $\overline{I_{[51]}}=I_{[51]}$ and hence represents another persistent root structure.
Like before, the action of the conformal group, through $GL(2,\mathbb{C})$, is transitively on the set of $[51]$ sextics.
Hence this corresponds to a single conformal class.
A suitable choice is the (Euclidean superintegrable) system
\begin{multline}
 V_{VII} = a (x_1+i x_2)+b \left(3 (x_1+i x_2)^2+x_3\right)\\+c\left(16 (x_1+i x_2)^3+(x_1-i x_2)+12 x_3 (x_1+i x_2)\right) \\
                       +d\left(5 (x_1+i x_2)^4+(x_1^2+x_2^2+x_3^2\right)+6 \left(x_1+i x_2)^2 x_3\right)+e,
\end{multline}
which has classifying sextic
\[
 p(z)=24i \left((x_1+ix_2)z-\frac{3}{2} \right)z^5.
\]
Note that this sextic cannot degenerate into a $[6]$-type structure at any finite point.

\paragraph{Case [42]:}The ideal of conditions for the $[42]$ root structure is generated by the coefficients of the 5 covariants
\begin{align}
 B_{ 0}^{[42]} &= 
       B_0,  \nonumber\\
 C_{ 1}^{[42]} &= 
       C_1,  \nonumber\\
 D_{ 8}^{[42]} &= 
       27 B_4^2-50 B_2A_3^2, \nonumber\\
 D_{ 6}^{[42]} &= 
       20 C_3 A_3 + B_4 B_2, \nonumber\\
 D_{ 0}^{[42]} &= 
       D_0.
\end{align}
Unlike the cases examined so far, this ideal is not closed under differentiation.
Adding the first derivatives closes the ideal, and using Gr\"obner bases it can be shown that 
\begin{align}
{I_{[6]}}^3& \subset \overline{I_{[42]}}\subset{I_{[6]}}.\label{eq:containment_of_the_closed_42_structure_and_the_6_structure}
\end{align}
From this, and the fact that $I_{[6]}$ is radical, it can be concluded that the radical of the differential closure of $I_{[42]}$ is $I_{[6]}$, i.e.\ 
\[
  \sqrt{\overline{I_{[42]}}}=I_{[6]}.  
\]

So, as should have been expected given the differential non-closure of the ideal $I_{[42]}$, forcing the conditions for the $[42]$-type root structure to hold identically will only yield sextics which are generically of type-$[6]$ (The type-$[0]$ root structure is also a possible degeneration, however this is just a further degeneration of type-$[6]$).
This proves that no system has coefficient functions that form a type-$[42]$ sextic everywhere.
However a sextic with a type-$[42]$ root structure gives valid values for the sextic for $p(z)$ and hence there must exist a system whose classifying sextic takes a type-$[42]$ root structure at a non-generic point.

\paragraph{Case [33]:}
The ideal $I_{[33]}$ of conditions for the $[33]$-type root structure is generated by the coefficients of the three covariants
\begin{align*}
 C_{ 6}^{[33]} &= 
       C_6,  \\
 C_{ 4}^{[33]} &= 
       C_4,  \\
 C_{ 3}^{[33]} &= 
       33 B_0 A_3 - 5 C_3.
\end{align*}
This ideal is closed under differentiation.
Since the only degenerations of the $[33]$-type root structure are either $[6]$-type or $[0]$-type it is safe to conclude that a system whose coefficient functions yield a $[33]$-type root structure at a regular point has the $[33]$-type root structure at every regular point.
The transitivity of $GL(2,\mathbb{C})$ on three or fewer roots means that this will correspond to a single conformal class.
A particular representative of the systems in this conformal class is given by the (Euclidean superintegrable) system
\[
  V_{OO} =a (4 x_1^2+4 x_2^2+x_3^2)+b x_1+c x_2+\frac{d}{x_3^2}+e
\]
which has classifying sextic
\[
 p(z)=\frac{6i}{x_3}z^3.
\]

\paragraph{Case [411]:}The ideal $I_{[411]}$ of conditions for the $[411]$ root structure is generated by the coefficients of the three covariants
\begin{align*}
 B_{ 0}^{[411]} &= 
       B_0, \\
 C_{ 1}^{[411]} &=  
       C_1, \\
 D_{ 0}^{[411]} &=  
       D_0. 
\end{align*}
The ideal $I_{[411]}$ is closed under differentiation.
So if a system has `structure functions' giving a $[411]$-type sextic a regular point, it will do so in an open set around that point.
The ideal $I_{[411]}$ is contained within the ideal $I_{[42]}$ (meaning the algebraic set defined by the ideal $I_{[411]}$ contains the algebraic set defined by the ideal $I_{[42]}$) and hence, even without explicitly checking, it's clear that the transient $[42]$-type root structure discussed above will break up into the $[411]$ structure under translation of the regular point.

The local action of the conformal group, through $GL(2,\mathbb{C})$, is transitive on the $[411]$-type sextics and hence every system in this class will be conformally related.
Note that, up to a local-conformal transformation, there are two distinct types of regular-points within system in this class, those which give a $[42]$-type sextic and those which give a $[411]$-type sextic.

A particular representative of the systems in this conformal class is given by the (Euclidean superintegrable) system
\[
  V_{V}=a (4 x_3^2+x_1^2+x_2^2)+b x_3+\frac{c}{(x_1+i x_2)^2}+d \frac{x_1-i x_2}{(x_1+i x_2)^3}+e,
\]
which has classifying sextic
\[
 p(z)=\left(\frac{9i}{x_1+ix_2}z^2+\frac{3i(x_1-i x_2)}{(x_1+i x_2)^2}z\right).
\]
This sextic is generically of type-$[411]$, but takes the $[42]$-type root structure along the hypersurface 
\[
 x_1-i x_2=0.
\]

\paragraph{Case [321]:}
The ideal $I_{[321]}$ of conditions for the type-$[321]$ root structure is given by the coefficients of the 5 covariants
\begin{align*}
 D_{ 0}^{[321]} &= 
       11 B_0^2-25 D_0,  \\
 E_{+1}^{[321]} &=  
       3 C_1 B_0-5 E_1, \\
 F_{+8}^{[321]} &=  
       75 (2610 D_2+827 B_2 B_0){A_3}^2-100(125 C_3 B_2+144 C_1 B_4) A_3\\ 
       &\quad+(3125 {C_4}^2+5184 {B_4}^2 B_0), \\
 F_{+6}^{[321]} &=  
       300(61 {B_0}^2-115 D_0){A_3}^2-20(5 C_1 B_2+22 C_3 B_0) A_3\\
                &\phantom{=}- B_4(7 B_2 B_0-270 D_2), \\
 F_{ 0}^{[321]} &=  
       41 {B_0}^3-75 D_0 B_0-125 F_0.
\end{align*}
This ideal is not closed under differentiation but it closes after 2 derivatives.
By the obvious transitivity of $GL(2,\mathbb{C})$ on the type-$[321]$ sextics, it is clear that the location of the three roots cannot be relevant to the vanishing of $\overline{I_{[321]}}$.
An educated guess would therefore be that the algebraic-set corresponding to $\overline{I_{[321]}}$ will be points satisfying either the $I_{[33]}$ ideal or the $I_{[51]}$ ideal.
That is, it should be expected that
\[
 \sqrt{\overline{I_{[321]}}}\overset{?}{=}I_{[51]}\cap I_{[33]}.
\]
The intersection $I_{[51]}$ and $I_{[33]}$ can be calculated by eliminating $t$ from the convex combination of ideals $(1-t)I_{[33]}+tI_{[51]}$ (e.g.\ see Hassett, chapter 4 \cite{hassett2007introduction}).
Denoting this ideal by $I_{[51]\wedge[33]}$, calculations show this elimination ideal to be generated by the coefficients of the two covariants
\begin{align*}
 C_{ 4}^{[51]\wedge[33]} &= 
        C_4, \\
 C_{ 3}^{[51]\wedge[33]} &= 
        33B_0 A_3-5 C_3.
\end{align*}
Straightforward Gr\"obner basis calculations now show 
\begin{align*}
  \left(I_{[51]}\cap I_{[33]}\right)^3&\subset \overline{I_{[321]}}\subset I_{[51]}\cap I_{[33]}.
\end{align*}
Hence, as predicted,
\[
 \sqrt{\overline{I_{[321]}}} = I_{[51]}\cap I_{[33]}.
\]
So any systems whose coefficient functions cause the ideals $I_{[321]}$ to vanish identically must lie in either the $[33]$ or $[51]$ classes (or their degenerations) and thus have already been classified above.

\paragraph{Case [222]:}
The ideal of conditions for the $[222]$-type root structure is given by the coefficients of the 5 covariants
\begin{align*}
 D_{ 8}^{[222]} &= 
        50 B_2 A_3^2-27 B_4^2, \\
 D_{ 6}^{[222]} &= 
        160 B_0 A_3^2-B_4 B_2-20 C_3 A_3, \\
 D_{ 4}^{[222]} &= 
        -3 B_4 B_0+25 C_1 A_3, \\
 D_{ 2}^{[222]} &= 
        B_2 B_0+90 D_2, \\
 D_{ 0}^{[222]} &= 
        43 B_0^2-75 D_0.
\end{align*}
This ideal is not closed under differentiation, but closes after one derivative.
Straightforward calculations show
\begin{align*}
            I_{[6]}^4&\subset  \overline{I_{[222]}}\subset I_{[6]}. 
\end{align*}
Hence $\sqrt{\overline{I_{[222]}}}=I_{[6]}$ and any systems with `structure functions' that cause the ideal $I_{[222]}$ to vanish identically are in the class of  $[6]$-type systems and therefore have already been classified above.

\paragraph{Case [3111]:}
The ideal $I_{[3111]}$ of conditions for the $[3111]$-type sextics is given by the coefficients of the three covariants
\begin{align*}
 D_{ 0}^{[3111]} &=
        11 B_0^2-25 D_0, \\
 E_{ 1}^{[3111]} &=
        3 C_1 B_0-5 E_1, \\
 F_{ 0}^{[3111]} &=
        8 B_0^3-125 F_0.
\end{align*}
This is easily shown to be a differentially closed ideal, so we can conclude that a $[3111]$-type root structure is stable under translation of the regular point.
Unlike previous cases the action of the local conformal group (acting as $GL(2,\mathbb{C})$) is not automatically transitive on the set of $[3111]$-type sextics.
To distinguish between different possible $[3111]$-type root structures an extra piece of informations is necessary: the multi-ratio of the roots.
Once the mulit-ratio of the four roots is specified the $[3111]$-type sextic is canonically defined (up to a local conformal transformation).

Denoting the 4 roots by $\mathbf{r}=(\eta_1,\eta_2,\eta_3,\eta_4)$ (where $\eta_4$ will be the triply repeated root) the multi-ratio can be defined by 
\begin{align}
 \lambda=\frac{(\eta_1-\eta_2)(\eta_3-\eta_4)}{(\eta_2-\eta_3)(\eta_4-\eta_1)}.
\end{align}
Assuming $Y_{+3}$ is non-zero (or performing a small rotation such that it is non-zero), define
\[
\mathbf{a}=\frac{1}{Y_{+3}}\left(Y_{-1},Y_{ 0},Y_{+1},Y_{+2}\right). 
\]
The function is $\mathbf{a}$ can be expressed as functions of the $\eta_i$ using Vieta's formulas, and hence the Jacobian $\frac{\partial\lambda}{\partial\mathbf{x}}$ can be calculated via
\begin{align}
 \frac{\partial\lambda}{\partial\mathbf{x}}
        &=\frac{\partial\lambda}{\partial\mathbf{r}}\frac{\partial \mathbf{r}}{\partial \mathbf{a}}\frac{\partial \mathbf{a}}{\partial \mathbf{x}} \nonumber \\
       &=\frac{\partial\lambda}{\partial\mathbf{x}}\left(\frac{\partial \mathbf{a}}{\partial \mathbf{r}}\right)^{-1}\frac{\partial \mathbf{a}}{\partial \mathbf{x}}.\label{eq:the_lambda_x_Jacobian}
\end{align}
Without loss of generality the roots can be assumed take the values
\[
  \mathbf{r}_0=\left(-1,0,1,\frac{1-\lambda}{1+\lambda}\right)
\]
at the regular point.
Substituting this into \eqref{eq:the_lambda_x_Jacobian} gives,
\begin{align}
 \left.\frac{\partial \lambda}{\partial x_1}\right|_{\mathbf{r}_0}&=\frac{-8i(1-\lambda+\lambda^2)(\lambda-1) \lambda}{27(1+\lambda)^3},\nonumber\\
 \left.\frac{\partial \lambda}{\partial x_2}\right|_{\mathbf{r}_0}&=\frac{-16(1-\lambda+\lambda^2)(\lambda-1) \lambda}{27(1+\lambda)^3},\nonumber\\
 \left.\frac{\partial \lambda}{\partial x_3}\right|_{\mathbf{r}_0}&=0.\label{eq:the_action_of_a_translation_on_the_multiration_in_the_3111_case}
\end{align}
These imply that the action of the translations is rank 1 for almost all values of $\lambda$.
The five possible exceptions correspond to points where $\lambda=0,1,\infty$ or $\exp(\pm i\pi/3)$, where the action is, to a first order approximation, rank 0. 

The cases $\lambda=0,1,\infty$ correspond to the degenerate root structure $[411]$ and so have already been examined.
The case $\lambda=\exp(\pm i\pi/3)$ will be considered below.

\paragraph{Subcase [3111] \& Multi-Ratio$=\exp\left({\pm\frac{i\pi}{3}}\right)$:}
The analysis above shows that the first-order changes in the multi-ratio at value $\lambda=\exp(\pm i \pi/3)$ are zero with respect to first order changes in $x_1,x_2,x_3$.
Higher-order changes in $x_1,x_2,x_3$ may move $\lambda$ away from this value and so to examine whether or not this really is a persistent feature the corresponding ideal will need to be generated.
In what follows this will be labeled the $[3111]+CR$ root structure.

Under the action of $GL(2,\mathbb{C})$, a canonical form of the sextic with the $[3111]+CR$ root structure is given by
\[
 p_3(z)=z^3(z^3-1).
\]
We now identify the coefficients of general sextic \eqref{the_polynomial_p_z} with the general fractional linear transformation of $p_3(z)$ via
\[
  p(z)=(c_3z+c_4)^6 p_3\left(\frac{c_1z+c_2 }{c_3z+c_4}\right).
\]
Calculating an elimination ideal with respect to the variable $c_1,c_2,c_3,c_4$ gives an ideal generated by the coefficients of the covariants
\begin{align}
 D_{ 4}^{[3111]+CR} &=
        3600 C_1 A_3+288 B_4 B_0-125 B_2^2, \nonumber\\
 D_{ 2}^{[3111]+CR} &=
        B_2 B_0-10 D_2,\nonumber \\
 D_{ 0}^{[3111]+CR} &= 
        11 B_0^2-25 D_0.\label{eq:the_generating_covariants_for_3111_CR}
\end{align}
This ideal is closed under differentiation and hence this case is actually a persistent one.
Up to local equivalence there is only one sextic with the $[3111]+CR$ structure, and hence this represents a single conformal class.
A particular representative of this class is the (Euclidean superintegrable) system
\begin{multline}
 V_{VI} =  a \left(x_3^2-2 (x_1-i x_2)^3+4 (x_1^2+x_2^2)\right)+b \left(2 x_1+2 i x_2-3 (x_1-i x_2)^2\right) \\
          +c \left(x_1-i x_2\right)+\frac{d}{x_3^2}+e 
\end{multline} 
which has classifying sextic
\[
 p(z)=3i \left(z^3+\frac{2}{x_3}\right)z^3.
\]

\paragraph{Subcase [3111] \& Multi-Ratio$\neq\exp\left({\pm\frac{i\pi}{3}}\right)$:}
Since the action of a translation on the multi-ratio $\lambda$ (c.f \eqref{eq:the_lambda_x_Jacobian}) is rank 1 everywhere on the connected set $\mathbb{C}^*\setminus\left\{0,1,\infty,\exp\left({\pm\frac{i\pi}{3}}\right)\right\}$ every point will lie in a single orbit under this action.
Hence this there is a single conformal class of systems with this structure.

A particular representative is given by the (Euclidean superintegrable) potential
\begin{multline}
 V_{II} =  a \left({x_1}^2+{x_2}^2+{x_3}^2\right)+b\frac{\left(x_1 - i x_2\right)}{\left(x_1 + i x_2\right)^3} 
          +c\frac{1}{\left(x_1 + i x_2\right)^2}+d\frac{1}{{x_3}^2}+e
\end{multline}
which has classifying sextic
\[
 p(z)=\left(\frac{6i}{x_3}z^3+\frac{9i}{x_1+i x_2}z^2-\frac{3i(-x_1+i x_2)}{(x_1+i x_2)^2}\right).
\]
If we use this sextic to calculate the covariants \eqref{eq:the_generating_covariants_for_3111_CR} we find
\[
 D_2^{[3111]+CR}\propto\frac{1}{x_3^2 (x_1+i x_2)^2},
\]
verifying that this class cannot also contain the $[3111]+CR$ root structure.

\paragraph{Case [2211]:}
The ideal $I_{[2211]}$ of conditions for the $[2211]$ root structure is generated by a single covariant
\begin{align*}
 G_{ 6}^{[2211]} &=
        50(10 F_3^{(2)}+2 D_3 B_0+55 F_3^{(1)}) A_3-4 (43 B_0^2-75 D_0) C_6+75E_2 B_4.
\end{align*}
This ideal is not closed under differentiation, but closes after 3 derivatives.
Straightforward calculations show that
\begin{align*}
 \left(I_{[411]} \cap I_{[33]}\right)^4&\subset \overline{I_{[2211]}}\subset I_{[411]} \cap I_{[33]}.
\end{align*}
Hence 
\[
 \sqrt{\overline{I_{[2211]}}}=I_{[411]} \cap I_{[33]}.
\]
 Any systems with coefficient functions that cause the polynomials in the $I_{[2211]}$ ideal to vanish identically are in the class of $[411]$ or $[33]$ systems and thus have been classified.

\paragraph{Case [21111]:}
As is well known, the ideal of conditions for the $[21111]$ root structure is generated by one condition, the discriminant.
In terms of the Hilbert basis, the ideal $I_{[21111]}$ is generated by the single covariant
\begin{align*}
 J_{ 0}^{[21111]} &=
        5393 B_0^5-20125 D_0 B_0^3+18750 D_0^2 B_0-31875 F_0 B_0^2+56250 F_0 D_0+28125 J_0.
\end{align*}
The ideal $I_{[21111]}$ is not closed under differentiation, but closes after five derivatives.
Straightforward calculations show
\begin{align*}
          \left(I_{[3111]}\right)^4&\subset \overline{I_{[21111]}}\subset I_{[3111]},
\end{align*}
and so
\[
 \sqrt{\overline{I_{[21111]}}}=I_{[3111]}.
\]
Hence any systems with coefficient functions that identically satisfies the $I_{[21111]}$ ideal are in the $[3111]$ class and have already been classified above.

\paragraph{Case [111111]:}
All systems corresponding to sextics with a persistent root of multiplicity two or greater have been classified above.
All that remains is to classify systems corresponding to sextics with six distinct roots.
To identify a [111111]-type sextic up to the action of $GL(2,\mathbb{C})$, three absolute invariants need to be known.
An obvious choice would be three independent multi-ratios.
For the following discussion these independent multi-ratios are chosen to be
\begin{align}
  \lambda_4&=\frac{(\eta_1-\eta_2)(\eta_3-\eta_4)}{(\eta_2-\eta_3)(\eta_4-\eta_1)}, \nonumber\\
  \lambda_5&=\frac{(\eta_1-\eta_2)(\eta_3-\eta_5)}{(\eta_2-\eta_3)(\eta_5-\eta_1)}, \nonumber\\
  \lambda_6&=\frac{(\eta_1-\eta_2)(\eta_3-\eta_6)}{(\eta_2-\eta_3)(\eta_6-\eta_1)}. \label{eq:three_fundamental_cross_ratios}
\end{align}
These multi-ratios are useful for describing the root structure geometrically.
However they are not particularly suitable for examining translation of the regular point.
For this we define an alternative set of absolute invariants.
These are achieved by balancing out the covariant weight of the one dimensional representations 
\begin{equation}
 \mathbf{I}=\left(\frac{D_0}{B_0^2},\frac{F_0}{B_0^3},\frac{J_0}{B_0^5}\right).\label{eq:three_absolute_invariants_constructed_from_the_coefficients}
\end{equation}
It is safe to assume that $B_0$ is non-zero as doing otherwise leads back to the case $[411]$.
Examining the absolute invariants in $\mathbf{I}$ should give equivalent results to examining the action of the multi-ratios provided the map between them is invertible.
Since $\mathbf{I}$ can be expressed as a function of the multi-ratios (as can all absolute invariants), we can calculate the Jacobian determinant 
\begin{align}
  \det\left(\frac{\partial \mathbf{I}}{\partial \bm{\lambda}}\right)
     &= \det\left(\frac{\partial\big(\frac{D_0}{B_0^2},\frac{F_0}{B_0^3},\frac{J_0}{B_0^5}\big)}{\partial(\lambda_4,\lambda_5,\lambda_6)}\right).
     \label{eq:the_Jacobian_between_the_two_types_of_absolute_invariants}
\end{align}
 The Jacobian determinant \eqref{eq:the_Jacobian_between_the_two_types_of_absolute_invariants} factors nicely and can be seen to vanish if and only if there is either a double root or if the condition
\begin{align}
\lambda_4-\lambda_5 \lambda_6&=0 \label{eq:a_particular_incarnation_of_the_multiratio_in_terms_of_the_three_fundamental_cross_ratios}
\end{align}
 is satisfied (up to permutation of roots).
 Condition \eqref{eq:a_particular_incarnation_of_the_multiratio_in_terms_of_the_three_fundamental_cross_ratios} is a well known object in the literature, going by the name of the $M_6=-1$ multi-ratio condition and is an interesting object of study in its own right \cite{king2003tetrahedra}.
 Written in terms of the roots, \eqref{eq:a_particular_incarnation_of_the_multiratio_in_terms_of_the_three_fundamental_cross_ratios} is equivalent to
\begin{equation}
 \frac{(\eta_1-\eta_2)(\eta_5-\eta_3)(\eta_4-\eta_6)}{(\eta_2-\eta_5)(\eta_3-\eta_4)(\eta_6-\eta_1)}=-1.
\end{equation}
 
 Assuming the condition \eqref{eq:a_particular_incarnation_of_the_multiratio_in_terms_of_the_three_fundamental_cross_ratios} is satisfied, the roots can be assumed to take the values
\[
\mathbf{r}_0=\left(-1,0,1,\frac{1-\lambda_5\lambda_6}{1+\lambda_5\lambda_6},\frac{1-\lambda_5}{1+\lambda_5},\frac{1-\lambda_6}{1+\lambda_6}\right).
\]
The action of a translation on the value of $\lambda_4-\lambda_5 \lambda_6$ is (to first order) given by
\begin{align}
 \left.\frac{\partial (\lambda_4-\lambda_5 \lambda_6)}{\partial x_1}\right|_{\mathbf{r}_0}&=0,\nonumber\\
 \left.\frac{\partial (\lambda_4-\lambda_5 \lambda_6)}{\partial x_2}\right|_{\mathbf{r}_0}&=\frac{8(\lambda_5 \lambda_6-1)(\lambda_5^2 \lambda_6^2-\lambda_5^2 \lambda_6-\lambda_5 \lambda_6^2+\lambda_5^2+\lambda_6^2-\lambda_5-\lambda_6+1)^2}{27(1+\lambda_5 \lambda_6)(\lambda_6-1)(\lambda_5-1)(\lambda_5+1)(\lambda_6+1)},\nonumber\\
 \left.\frac{\partial (\lambda_4-\lambda_5 \lambda_6)}{\partial x_3}\right|_{\mathbf{r}_0}&=\frac{-8i(\lambda_5^2 \lambda_6^2-\lambda_5^2 \lambda_6-\lambda_5 \lambda_6^2+\lambda_5^2+\lambda_6^2-\lambda_5-\lambda_6+1)^2}{27(\lambda_6-1) (\lambda_5-1) (\lambda_5+1) (\lambda_6+1)}.
\end{align}
And so, remembering that all roots must be distinct, this action will only be rank zero if the condition
\begin{equation}
 \lambda_5^2 \lambda_6^2-\lambda_5^2 \lambda_6-\lambda_5 \lambda_6^2+\lambda_5^2+\lambda_6^2-\lambda_5-\lambda_6+1=0 \label{eq:the_additional_condition_for_the_persistant_multiratio_condition}
\end{equation}
is also satisfied.
Likewise, calculating the action of the derivative on \eqref{eq:the_additional_condition_for_the_persistant_multiratio_condition} shows that the action is (to first order) rank zero on \eqref{eq:the_additional_condition_for_the_persistant_multiratio_condition}.
So this is a promising candidate for a persistent feature.

\paragraph{Subcase [111111]; $M_6=-1$ \& $CR$:}
The ideal $I_{M_6+CR}$ of covariants vanishing under conditions \eqref{eq:a_particular_incarnation_of_the_multiratio_in_terms_of_the_three_fundamental_cross_ratios} and \eqref{eq:the_additional_condition_for_the_persistant_multiratio_condition}, can be calculated without too much effort.
This ideal is generated by the coefficients of the two covariants
\begin{align}
 F_{ 4}^{M_6+CR} &= 
        360 (49 C_1 B_0-48 E_1) A_3-193 B_2^2 B_0 -1896 C_3 C_1\nonumber \\
        &\phantom{=}\quad+288 D_0 B_4+3276 D_2 B_2, \nonumber \\
 F_{ 0}^{M_6+CR} &= 
        97 B_0^3-275 D_0 B_0+375 F_0.
\end{align}
Calculations show that the ideal $I_{M_6+CR}$ is closed under differentiation and hence represents a persistent feature.

To make conclusions about whether the bulk of the ideal $I_{M_6+CR}$, by which we mean the part of the algebraic variety for which the action of the full conformal group has maximal rank, corresponds a single conformal class requires us to understand the geometry of the corresponding algebraic set.
One could imagine a situation where the algebraic variety naturally forms two or more unconnected components, or where a component is disconnected when we remove points where the rank of the action is submaximal.
Thankfully such a situation would show up algebraically due to the following theorem (taken from Corollary~4.16 in Ref.~\cite{mumford1976algebraic}).

\begin{theorem}\label{thm:the_connectedness_of_complex_algebraic_varieties}
 Let $X\subset\mathbb{P}^n$ be an r-dimensional projective variety and let {$Y\subsetneqq{X}$} be a closed algebraic set.
 Then $X\setminus Y$ is connected in the classical topology.
\end{theorem}

Since the rank of the action can be examined using polynomial conditions (i.e.\ determinants and sub-determinants) the locations where the rank drops form an algebraic subset of the algebraic set defined by \eqref{eq:the_additional_condition_for_the_persistant_multiratio_condition}.
Maple's polynomial irreducibility test indicates that condition \eqref{eq:the_additional_condition_for_the_persistant_multiratio_condition} is absolutely irreducible over $\mathbb{C}$, and so the algebraic set defined by \eqref{eq:the_additional_condition_for_the_persistant_multiratio_condition} is actually an algebraic variety.
Theorem~\ref{thm:the_connectedness_of_complex_algebraic_varieties} implies that the set of points for which the action has maximal rank will be connected.
So this action will be transitive on set of point satisfying
\eqref{eq:a_particular_incarnation_of_the_multiratio_in_terms_of_the_three_fundamental_cross_ratios} and \eqref{eq:the_additional_condition_for_the_persistant_multiratio_condition} provided the action has rank 1.
The conditions on the multi-ratios where the rank of the action drops to zero can be shown to correspond to roots of multiplicity 2 (or higher), and by the results of the [21111] case, these must be transient conditions, i.e.\ they cannot hold identically in this case.

So the ideal $I_{M_6+CR}$ defines an irreducible variety, corresponding to single conformal class.
A particular representative of the systems lying in this class is given by the (Euclidean superintegrable) system
\begin{align}
 V_{IV} = a (4 x_1^2+x_2^2+x_3^2)+b x_1+\frac{c}{x_2^2}+\frac{d}{x_3^2}+e 
\end{align}
which has classifying sextic
\[
 p(z)=\frac{3}{4x_2}(z^2+1)^3+\frac{6i}{x_3}z^3.
\]

\paragraph{Subcase [111111]; Rank 3 Jacobian:}
Assuming now that the multi-ratio condition \eqref{eq:a_particular_incarnation_of_the_multiratio_in_terms_of_the_three_fundamental_cross_ratios} is not satisfied, the Jacobian \eqref{eq:the_Jacobian_between_the_two_types_of_absolute_invariants} will be rank 3.
The rank of the action on the absolute invariants \eqref{eq:three_absolute_invariants_constructed_from_the_coefficients} can be used to examine the rank of the action on the multi-ratios \eqref{eq:three_fundamental_cross_ratios}.

Calculating the determinant of the Jacobian between the three absolute invariants \eqref{eq:three_absolute_invariants_constructed_from_the_coefficients} and the coordinates gives,
\begin{align}
 \det\left(\frac{\partial \mathbf{I}}{\partial\mathbf{x}}\right) 
           &=\det\left( \frac{\partial\big(\frac{D_0}{B_0^2},\frac{F_0}{B_0^3},\frac{J_0}{B_0^5}\big)}{\partial(x_1,x_2,x_3)}\right)  \nonumber\\
           &=\frac{\left(\begin{array}{l}2521 B_0^5-9625 D_0 B_0^3+6250 D_0^2 B_0\\ \qquad-7500 F_0 B_0^2 +65625 F_0 D_0-84375 J_0\end{array}\right)}{2^5 3^{10} 5^6 B_0^{11}}O_{ 0}.\label{the_Jacobian}
\end{align}
This shows that action is rank 3 away from 
\begin{align}
  J_0^{Jac} &= 2521 B_0^5-9625 D_0 B_0^3+6250 D_0^2 B_0-7500 F_0 B_0^2+65625 F_0 D_0-84375 J_0\nonumber\\
           &= 0 \label{J0_covariant_rank_2_condition} 
\end{align}
and 
\begin{align}
 O_{ 0} = 0.
\end{align}
A careful examination reveals that $O_0$ is a symmetric version of the aforementioned $M_6=-1$ multiratio condition.
Specifically
\begin{align}
 O_{ 0}\propto \prod_{\sigma\in \Sigma} \left(\begin{array}{l} 
                                                     (r_{\sigma(1)}-r_{\sigma(2)})(r_{\sigma(3)}-r_{\sigma(4)})(r_{\sigma(5)}-r_{\sigma(6)}) \\
                                              \qquad+(r_{\sigma(6)}-r_{\sigma(2)})(r_{\sigma(2)}-r_{\sigma(3)})(r_{\sigma(4)}-r_{\sigma(5)})
                                        \end{array}\right),
\end{align}
where $\Sigma$ is the fifteen elements of the permutation group that give the fifteen different versions of the $M_6=-1$ condition. 
By the discussion in the previous section, systems satisfying $O_{ 0}=0$ have already been considered.
So henceforth $O_{ 0}$ will be assumed non-zero.

Returning to the main argument, the action of the translations will be rank~3 on the absolute invariants $\mathbf{I}$ away from $J^{Jac}_0=0$ and $O_0 = 0$ and by theorem~\ref{thm:the_connectedness_of_complex_algebraic_varieties} this is connected.
Because the action is rank three on this connected, three-dimensional set, there can only be one orbit under this action.

 A particular representative of the systems in this orbit is given by the conformally-superintegrable potential
\[
 V_{S}=\frac{a}{(1+x_1^2+x_2^2+x_3^2)^2}+\frac{b}{x_1^2}+\frac{c}{x_2^2}+\frac{d}{x_3^2}+\frac{e}{(-1+x_1^2+x_2^2+x_3^2)^2}.
\]
The potential $V_{S}$ above is St\"ackel equivalent to the superintegrable potential
\[
 V_{S^\prime}=\frac{\alpha}{s_1^2}+\frac{\beta}{s_2^2}+\frac{\gamma}{s_3^2}+\frac{\delta}{s_4^2}+\epsilon
\]
defined over the 3-sphere
\[
 s_1^2+s_2^2+s_3^2+s_4^2=1.
\]

\paragraph{Subcase [111111]; Rank 2 Jacobian; $J^{Jac}_0\equiv0$:} Taking the condition $J^{Jac}_0$ on its own generates an ideal, which will be denoted $I_{Jac}$.
The differential closure of $I_{Jac}$ is easily seen by considering the following relations
\begin{align}
 \frac{\partial J^{Jac}_0}{\partial x}
      &=5 \left(X_{-1}-X_{+1}\right)J^{Jac}_0, \nonumber\\
 \frac{\partial J^{Jac}_0}{\partial y}
      &=5i\left(X_{-1}+X_{+1}\right)J^{Jac}_0, \nonumber\\
 \frac{\partial J^{Jac}_0}{\partial z}
      &=5\sqrt{2}\left(X_{ 0}\right)J^{Jac}_0.
\end{align}
The Hilbert dimension of the ideal $I_{Jac}$ is 6 and, a check for absolutely irreducibility using Maple returns a positive result.
This means (like $I_{M_6+CR}$) the ideal  $I_{Jac}$ gives an algebraic variety and hence is a connected set.
The local action of the conformal group is rank 4 and the action of translation of the regular point on the absolute invariants is rank 2 (which is clearly distinct from the local action).
Hence the generic action on this space will be rank 6 and thus can only be one orbit under the action of the conformal group.

So $I_{Jac}$ represents a single conformal class and a particular representative of the systems in this class is given by
\begin{align}
 V_{I} = a ({x_1}^2+{x_2}^2+{x_3}^2)+\frac{b}{{x_1}^2}+\frac{c}{{x_2}^2}+\frac{d}{{x_3}^2}+e\label{Potential_I}
\end{align}
which has the classifying sextic
\[
 p(z)=\frac{6i}{x_3}z^3+\frac{3}{4x_2}(1+z^2)^3-\frac{3i}{4x_1} (1-z^2)^3.
\]
\paragraph{Subcase [111111]; Rank 1 Jacobian:}
Any further restrictions would necessarily show up when the Jacobian is rank 1.
Examining the $2\times2$ subminors of the Jacobian under the restriction $J_0^{Jac}=0$ gives an additional $14$th order covariant that must vanish identically, namely
\begin{align}
N_1^{Rank1}
    &= (125 F_0 +49 B_0^3-125 D_0 B_0 ) H_1 -20 (-25 D_0 +14 B_0^2) J_1 +150 L_1 B_0.
\end{align}
The ideal $I_{Rank1}$ generated by the coefficients of the covariants $N_1^{Rank1},J_0^{Jac}$ closes after 3 derivatives.
Straightforward calculations show
\begin{align*}
 \left(I_{M_6+CR} \cap I_{[3111]}\right)^3&\subset \overline{I_{Rank1}}\subset I_{M_6+CR} \cap I_{[3111]}.
\end{align*}
Hence $\sqrt{\overline{I_{Rank1}}}=I_{M_6+CR} \cap I_{[3111]}$ and all corresponding systems have already been classified.

This completes the classification. There are a total of 10 conformal classes.
A given maximal-parameter, second-order conformally-superintegrable system can be placed into one of the aforementioned conformal classes by identifying which of the ideals above identically vanish.
Table~\ref{table:relative_invariants_for_the_conformally_superintegrable_potentials} shows the pattern of vanishing ideals for each of the representative systems and hence for every system in that class.

\begin{table}
\centering
\renewcommand{\arraystretch}{1.2}
{
\begin{tabular}{|c|c|c|c|c|c|c|c|c|c|c|c|}
\hline
   & $I_{[0]}$ & $I_{[6]}$ & $I_{[51]}$ & $I_{[411]}$ & $I_{[33]}$ & $I_{[3111]}$ & $I_{[3111]+CR}$ & $I_{M_6+CR}$ & $I_{Jac}$ \\
\hline
\hline
$S$	&   &   &   &   &   &   &   &   &   \\
$I$	&   &   &   &   &   &   &   &   & 0 \\
$II$	&   &   &   &   &   & 0 &   &   & 0 \\
$IV$	&   &   &   &   &   &   &   & 0 & 0 \\
$V$	&   &   &   & 0 &   & 0 &   & 0 & 0 \\
$VI$	&   &   &   &   &   & 0 & 0 & 0 & 0 \\
$VII$	&   &   & 0 & 0 &   & 0 & 0 & 0 & 0 \\
$O$	& 0 & 0 & 0 & 0 & 0 & 0 & 0 & 0 & 0 \\
$OO$	&   &   &   &   & 0 & 0 & 0 & 0 & 0 \\
$A$	&   & 0 & 0 & 0 & 0 & 0 & 0 & 0 & 0 \\
\hline
\end{tabular}}
\caption{Vanishing irreducible ideals for the ten maximum-parameter systems}\label{table:relative_invariants_for_the_conformally_superintegrable_potentials}
\end{table}

\section{Conclusion and Future Directions}

The results obtained allow us to confidently state the systems listed elsewhere in the literature \cite{kalnins2006classification,kalnins2008fine,kalnins2011laplace} form a complete list of St\"ackel equivalent 3-dimensional, maximum-parameter, superintegrable potentials over conformally-flat spaces.
By considering specialisations of the potentials above that give St\"ackel multipliers for a flat metric, we reprove the result of Ref.~\cite{kalnins2007algebraicvarieties3d}.
There are 10 superintegrable systems over Euclidean space, nine of these flat-space systems appear explicitly above and the tenth one is St\"ackel equivalent to system $V$, as given in Ref.~\cite{kalnins2007algebraicvarieties3d}.
Likewise, by considering those St\"ackel multipliers which give a non-zero constant curvature metric it can be shown there are six systems on the 3-sphere, St\"ackel equivalent to $S,I,II,IV,VI,$ and $OO$.
This proves that the list of known maximum-parameter systems is complete.

The ten conformal classes given above have a natural partial ordering, obtained by considering ideal containment for the classifying polynomial ideals.
This partial order suggests a way to think of obtaining one conformal class as the limit of another.
For example, the ideal $I_{[3111]+CR}$ contains the ideal $I_{[3111]}$ as a subideal, meaning the algebraic variety defined by $I_{[3111]}$ contains the algebraic variety defined by $I_{[3111]+CR}$.
A particular point in a $[3111]$-type system corresponds to a generic point in the variety defined by $I_{[3111]}$.
A conformal motion can then move this point in the variety and take it arbitrarily close to $[3111]+CR$ subvariety.
This partial ordering is displayed in figure~\ref{figure:detailed_ideal_containment_diagram}.
The arrows  point from subideal to superideal.
Ideal constainment for any of the ideals provided can be decided by determining whether there is a directed path between the ideals shown in figure~\ref{figure:detailed_ideal_containment_diagram}.
If instead we wish to think of this digrams in terms of subvarieties, i.e.\ thinking about limiting from one variety to another, the direction of the arrows should be reversed.

There are three pieces of information contained in the boxes in figure~\ref{figure:detailed_ideal_containment_diagram}, the first is the name of the chosen representative of the system in the classification above, the second is the factor structure for the sextic associated to the system, the precise details of which can be found above, the third is a reference to the bracket notation of B{\^o}cher, where a partition of five indicates the generic separable coordinates in which the system separates \cite{bocher1894reihenentwickelungen}, and the last piece of information is the Hilbert dimension of the ideal in the variables $Y_i$.
As should be expected, almost all of the minimal degenerations shown in figure~\ref{figure:detailed_ideal_containment_diagram} reduce of the Hilbert dimension by one. The only exception to this is the degeneration from type-$[6]$ sextic to type-$[0]$ sextic, where the two degrees of freedom coming from the position of the single root and the value of the leading coefficients are simultaneously lost.

The systems denoted by $O,OO$ and $A$ only separate in non-generic separable coordinates and so do not have a B\^ocher bracket associated to them.
However for the 7 classes that do, the partial-ordering of the B\^ocher brackets (as partitions of five) is the same as the partial-ordering given by the ideal containment relations.

When considering the St\"ackel equivalent systems, this limiting scheme should help construct a contraction scheme for their quadratic algebras which can potentially be used to provide an Askey-Wilson-type scheme for the special functions that arise in the two-parameter models of the quantum quadratic-algebras.

\begin{figure}[h!]
   \centering
\begin{displaymath}
\xymatrix{
                                                                                                                                               & { \begin{tabular}{|c|} \hline $\vphantom{\Big(}\mathbf{O}$   \\ $[0]$         \\ $n/a$     \\ $d=0$ \\ \hline \end{tabular}}\ar[dr]        &                                                                                                              \\
                                                                                                                                               & { \begin{tabular}{|c|} \hline $\vphantom{\Big(}\mathbf{VII}$ \\ $[51]$        \\ $(5)$     \\ $d=3$ \\ \hline \end{tabular}}\ar[d]\ar[dl]  & { \begin{tabular}{|c|} \hline $\vphantom{\Big(}\mathbf{A}$ \\ $[6]$  \\ $n/a$ \\ $d=2$ \\ \hline \end{tabular}}\ar[d]\ar[l] \\
  { \begin{tabular}{|c|} \hline {$\vphantom{\Big(}\mathbf{III}$}\\ $[411]$    \\ $(23)$  \\ $d=4$ \\ \hline \end{tabular}}\ar[d]\ar[dr]    & { \begin{tabular}{|c|} \hline $\vphantom{\Big(}\mathbf{VI} $ \\ $[3111]+CR$   \\ $(41)$    \\ $d=4$ \\ \hline \end{tabular}}\ar[d]\ar[dl] & { \begin{tabular}{|c|} \hline $\vphantom{\Big(}\mathbf{OO}$\\ $[33]$ \\ $n/a$ \\ $d=3$ \\ \hline \end{tabular}}\ar[l]        \\
  { \begin{tabular}{|c|} \hline $\vphantom{\Big(}\mathbf{IV}$ \\ $M_6=-1$ \\ $(311)$ \\ $d=5$ \\ \hline \end{tabular}}\ar[dr]              & { \begin{tabular}{|c|} \hline $\vphantom{\Big(}\mathbf{II}$  \\ $[3111]$      \\ $(221)$   \\ $d=5$ \\ \hline \end{tabular}}\ar[d]         &                                                                                                              \\
                                                                                                                                               & { \begin{tabular}{|c|} \hline $\vphantom{\Big(}\mathbf{I}  $ \\ $J^{Jac}_0=0$ \\ $(2111)$  \\ $d=6$ \\ \hline \end{tabular}}\ar[d]         &                                                                                                              \\
                                                                                                                                               & { \begin{tabular}{|c|} \hline $\vphantom{\Big(}\mathbf{S} $  \\ $[111111]$    \\ $(11111)$ \\ $d=7$ \\ \hline \end{tabular}}               & { \begin{tabular}{ rl } \hline Key:& \textbf{System Name} \\ & Factor Structure \\ & B{\^o}cher Bracket \\ & Hilbert Dimension $(d)$ \\ \hline \end{tabular}}
}
\end{displaymath}
\caption{Subvariety Hasse Diagram}\label{figure:detailed_ideal_containment_diagram}
\end{figure}
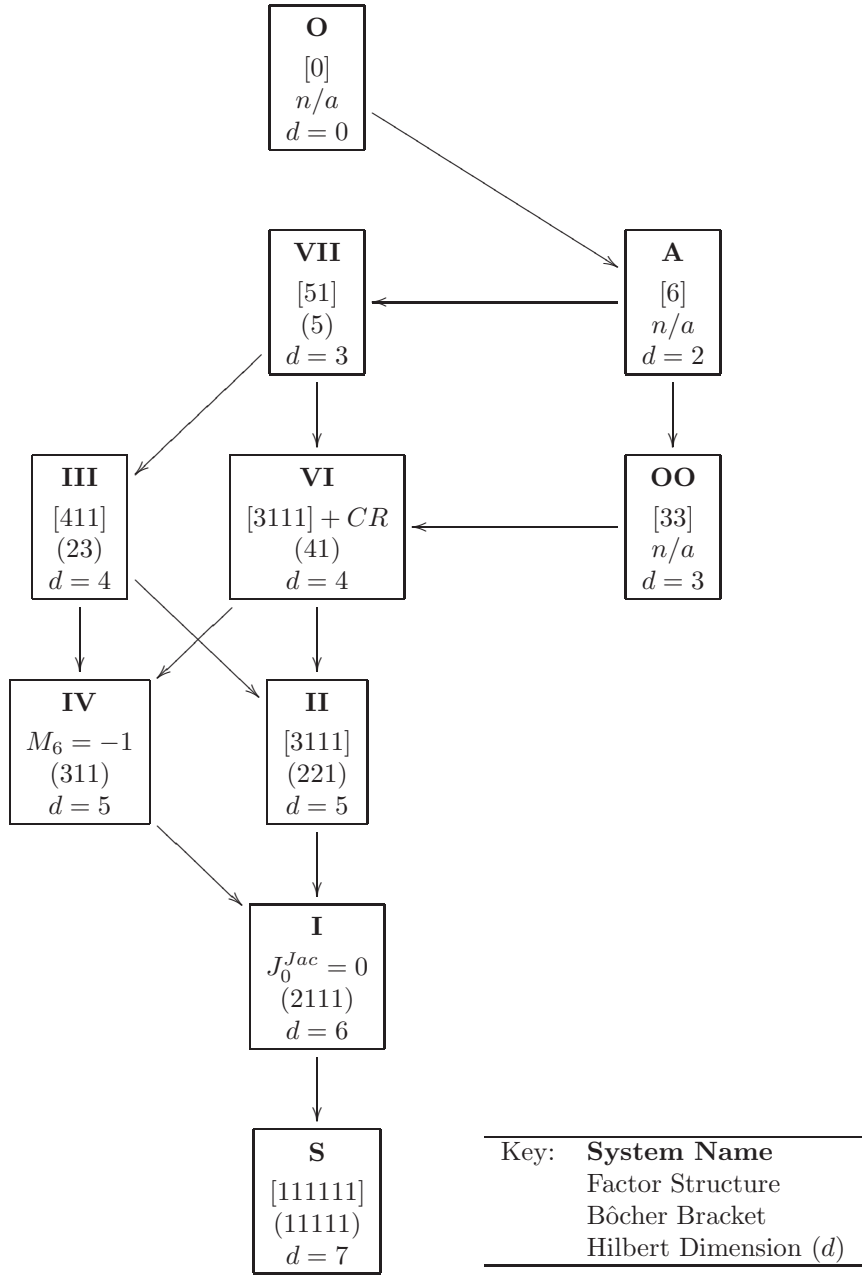

\bibliography{./papers}
\bibliographystyle{unsrt}

\end{document}